\documentclass[10pt]{IEEEtran}
\onecolumn
\pdfoutput=1

\usepackage{amsthm}
\usepackage{amsfonts}
\usepackage{dsfont}
\usepackage{mathrsfs}
\usepackage{enumerate}
\usepackage{amsmath}
\usepackage{subfigure}
\usepackage{graphicx}
\usepackage{color}
\usepackage{amsmath}
\usepackage{algpseudocode}
\usepackage{algorithm}

\usepackage{amssymb}
\usepackage{amstext}

\newtheorem{theorem}{Theorem}

\newtheorem{lemma}{Lemma}
\newtheorem{note}{Note}

\newtheorem{definition}{Definition}

\newcommand{\field}[1]{\mathbb{#1}}
\newcommand{\N}{\field{N}} 
\newcommand{\R}{\field{R}} 
\newcommand{\Z}{\field{Z}} 
\newcommand{\1}{{\bf 1}} 
\newcommand{\E}{\, \mathsf{E}} 

\newcommand{\Dscr}{{\cal D}}

\newcommand{\Nscr}{{\cal N}}

\newcommand{\Rscr}{{\cal R}}
\newcommand{\Sscr}{{\cal S}}

\newcommand{\bea}{\begin{eqnarray}}
\newcommand{\eea}{\end{eqnarray}}
\newcommand{\beas}{\begin{eqnarray*}}
\newcommand{\eeas}{\end{eqnarray*}}

\newcommand{\eq}[1] {(\ref{#1})}

\newcommand{\dist}{\phi}
\newcommand{\qd}{\xi}

\newcommand{\bv}[1]{{\bf #1}}

\newcommand{\til}{\tilde}

\newcommand{\ep}{\epsilon}
\newcommand{\vep}{\varepsilon}
\newcommand{\C}{{\mathcal C}}

\newcommand{\D}{{\mathcal D}}

\newcommand{\one}{{\mathds 1}}

\newcommand{\defi}{\triangleq}

\newcommand{\goesto}{\rightarrow}

\newcounter{constcount}
\newcounter{numcount}
\newcommand{\rescnt}{\setcounter{numcount}{1}}

\newcommand{\DFT}{\text{DFT}}

\newcommand{\aln}[1]{\begin{align*}#1\end{align*}}

\newcommand{\al}[1]{\begin{align}#1\end{align}}

\renewcommand{\vec}[1]{{\bf #1}}

\title{Network Compression: Worst-Case Analysis}

\author{Himanshu Asnani, Ilan Shomorony, A. Salman Avestimehr, and Tsachy Weissman\thanks{A shorter version of this paper was submitted to the International Symposium on Information Theory 2013}}

\begin{document}
\maketitle

\begin{abstract}

We study the problem of communicating a distributed correlated memoryless source over a memoryless network, from source nodes to destination nodes, under quadratic distortion constraints. 
We establish the following two complementary results: (a) for an arbitrary memoryless network, among all distributed memoryless sources of a given correlation, Gaussian sources are least compressible, that is, they admit the smallest set of achievable distortion tuples, and (b) for any memoryless source to be communicated over a memoryless additive-noise network, among all noise processes of a given correlation, Gaussian noise admits the smallest achievable set of distortion tuples. 
We establish these results constructively by showing how schemes for the corresponding Gaussian problems can be applied to achieve similar performance for (source or noise) distributions that are not necessarily Gaussian but have the same covariance.   

\end{abstract}

\section{introduction}
\label{sec::intro}

Stochastic modeling of the data source and the communication medium are essential in data compression and data communication problems. However, extracting these descriptions from a practical system is in general difficult and often leads to intractable problems from a theoretical point of view. As a result, Gaussian models for both the data sources and the noise in communication networks prevail.

The modeling of the noise in communication links as additive Gaussian is generally justified through the Central Limit Theorem, which suggests that the cumulative effect of many independent noise sources should be approximately Gaussian. The modeling of data sources as Gaussian, on the other hand, is less justifiable and done largely for the sake of analytical tractability.

From a theoretical standpoint, one way of supporting the Gaussian assumption is by establishing that it is worst-case, meaning that, within a given family of distributions (usually defined by a covariance constraint), the Gaussian assumption results in the smallest possible capacity or rate-distortion region.
In fact, this has long been known to be the case in two classical single-user Information Theory scenarios.
In the channel coding setting, it is known that, given a fixed variance of the noise, the Gaussian distribution minimizes the capacity of a memoryless additive-noise channel. The source coding counterpart of this result is that, for a fixed-variance i.i.d. random source, the Gaussian distribution minimizes the rate-distortion region. Both of these assertions can be proved using the fact that, subject to a variance constraint, the Gaussian distribution maximizes the entropy. In the channel coding case, a more operational proof of the fact that Gaussian noise is the worst-case noise was provided in  \cite{LapidothNearest}, where it was shown that random Gaussian codebooks and nearest-neighbor decoding achieve the capacity of the corresponding AWGN channel on a non-Gaussian channel. 

There are a few other worst-case characterizations in the literature. One example is \cite{DiggaviWorstCase}, where the authors consider vector channels with additive noise subject to the constraint that the noise covariance matrix lies in a convex set.
It is shown that, in this setting, the worst-case noise is vector Gaussian with a covariance matrix that depends on the transmit power constraints.
In \cite{ShamaiWorstCase}, a scalar additive-noise channel with binary input is considered. 
In this setting, the probability mass function of the (discrete) worst-case noise is characterized, and the worst-case capacity (i.e., the capacity under the worst-case noise) is found. Another example is the work of \cite{Aaron2Terminal} that characterizes the rate-distortion region for the two-encoder source coding problem with quadratic distortion constraints and Gaussian sources, which in turn allows the characterization of the joint Gaussian source as the worst-case source for the two-encoder quadratic source coding problem.

Beyond the aforementioned examples, worst-case analysis of  more general multi-user networks was, until recently, fairly limited. The main challenge lay in the fact that most multi-user Information Theory problems remain unsolved, i.e., without an explicit characterization of the capacity or rate-distortion regions. Recently, a new approach was introduced in \cite{wcnoisefull} that allowed to generalize the worst-case noise result from additive-noise point-to-point channels to arbitrary linear additive-noise wireless networks\footnote{In these networks, the received signal at each node is a linear combination of the transmit signals at all other nodes plus a noise term.}. The framework in \cite{wcnoisefull} can be described in two main steps. 
First, a  DFT (Discrete Fourier Transform)-based linear transformation is applied to all transmitted and received signals in the network in order to create an effective network where the additive-noise terms are ``approximately Gaussian''\footnote{In the sense that their distribution converges to a Gaussian distribution as the size of the blocks to which we apply the DFT-based transformation increases.}.  Next, by demonstrating the optimality of  coding schemes with finite precision\footnote{In coding schemes with finite precision, the encoding and decoding operations of the nodes in the network may only take inputs with a finite decimal expansion. This precision can become arbitrary large as the coding block length increases.} in Gaussian networks, it is proven that the capacity region of the Gaussian network is contained in the capacity region of the effective network asymptotically (as the size of the blocks to which we apply the DFT-based transformation increases). 
This approach was later utilized in~\cite{wcsourceITW} to establish that Gaussian sources are worst-case data sources for distributed compression of correlated sources over rate-constrained, noiseless channels, with a quadratic distortion measure (i.e., in the context of the quadratic $k$-encoder source coding problem).

In this work, we pursue the analogue of these worst-case results in joint source-channel coding, by considering the problem of distributed compression of information over an arbitrary network. 
More precisely, $k$ nodes in the network have access to correlated stochastic sources and wish to transmit them over an $N$-node network to respective destinations. A coding scheme is employed to define the encoding, relaying and decoding operations of the network nodes, and its performance metric is the mean square error in the destinations' reconstruction of their desired sources. This problem lies at the heart of increasingly many applications concerning distributed compression of information over a network, such as sensor networks. 

Since this setup involves the modeling of both the sources and the network, the worst-case characterization takes the form of two related sub-questions:
\begin{itemize}
\item \textbf{Question 1}: Given an arbitrary memoryless network, for
a fixed correlation amongst its distributed memoryless components, are the jointly Gaussian sources, the worst compressible? 
In other words, do they have the smallest set of achievable distortion tuples?
\item \textbf{Question 2}: Given an arbitrary memoryless distributed source, for an
additive-noise network with a given noise correlation, is the Gaussian noise worst-case,
in the sense of having the smallest set of achievable
distortion tuples?
\end{itemize} 

In this paper, we answer both of these questions in the affirmative. We utilize the aforementioned framework to propose a universal way of converting a coding scheme designed under the Gaussian assumption into coding schemes that can handle and attain similar performances for non-Gaussian sources or noises. In particular, we start by using the DFT-based linear transformation as a way to make either the sources or the noises approximately Gaussian. Since this operation introduces a statistical dependence between the resulting sources or noises, an interleaving scheme is employed, in order to create blocks of i.i.d. approximately Gaussian sources and noises. Within each of the resulting blocks we then apply the original coding scheme designed under Gaussian models. We show that such a scheme, when performed over sufficiently long blocks, can achieve distortions arbitrarily close to those achieved by the original coding scheme designed for Gaussian sources or noises. This is done by showing that our original scheme can be assumed without loss of generality to satisfy two properties: finite precision\footnotemark[\value{footnote}] and bounded  outputs\footnote{In a coding scheme with bounded outputs, each component of the source reconstruction sequences produced by the destinations cannot exceed a given number $M$.}. These properties allow us to use standard tools regarding the convergence of random variables, such as the Dominated Convergence Theorem, to bound the distortion attained by the new coding scheme constructed based on the DFT-based linear transformation.

Our contribution lies not only in answering the above two questions in the affirmative and showing the worst-case nature of Gaussian assumptions, but also in describing a systematic way of converting coding schemes designed under Gaussian assumptions into coding schemes that can handle non-Gaussian assumptions. 
The idea behind the construction of such schemes is simple conceptually, using DFT-based linear transformations, which renders them also algorithmically tractable.

The rest of the paper is organized as follows. Section \ref{sec::problem}
presents the formal problem formulation along with the main results of the
paper.  An overview of the main ingredients used in the proofs of the main results is
provided in Section \ref{sec::preresults}.  Section \ref{sec::worstsource} studies the problem of finding the worst case source, given a fixed correlation
matrix, for compression over a given memoryless network while the worst case
nature of Gaussian additive noise, again for a fixed correlation structure, is proven for an arbitrary memoryless distributed source in Section \ref{sec::worstnoise}.  The paper is concluded in Section
\ref{sec::conclusion}.

\section{Problem Formulation and Main Results}
\label{sec::problem}

We are given a (stochastic) network, where source nodes want to communicate correlated memoryless sources across the network to respective destination nodes, subject to a distortion constraint. 
As outlined in Section \ref{sec::intro}, we address two complementary questions in this paper. First, we consider characterizing, for a fixed network, the worst-case source distribution; i.e., the source distribution for which the set of achievable distortion constraints is smallest.
In order to make this question meaningful, we fix the covariance of the joint distribution of the sources.
Second, we consider fixing the source distribution and asking what is the worst-case noise in the network.
To make the latter problem well-posed, we focus on additive-noise networks, where the covariance matrix of the noise terms is fixed. 

In order to formally state these two problems, we will need the following notation.
We refer to an $n$-tuple $\{X[t]\}_{t=0}^{n-1}$ by both $X^n$ and $\mathbf{X}$
(when the size of the tuple $n$ is clear from the context).  
If a random variable $X$ has a probability density function, it is denoted as $f_X(x)$, and if the conditional distribution of $X$ given $Y$  has a conditional probability density function, it is denoted as $f_{X|Y}(x|y)$.
The notation $[0:k]$ is shorthand for the set of natural numbers $\{0,1,\ldots,k\}$, and $X_i[0:k] = \{ X_i[0], X_i[1],\ldots, X_i[k] \}$.

\begin{figure}[htbp] 	
\begin{center}
\scalebox{0.5}{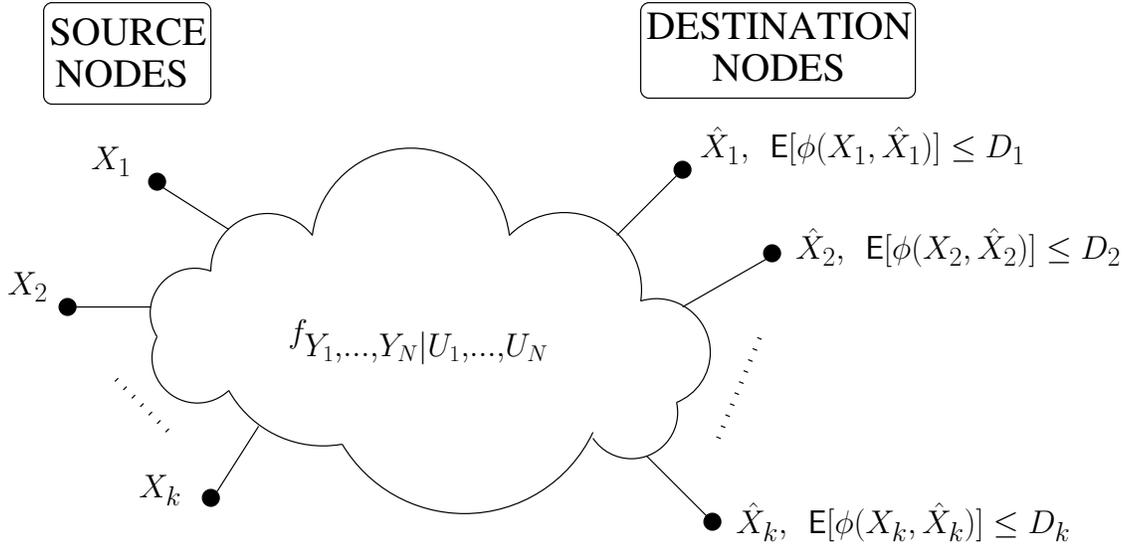}
\caption{$(k,N)$-memoryless network. $\phi(\cdot,\cdot)$ refers to the squared error.} 
\label{dm_mn}
\end{center}
\end{figure}

A $(k,N)$-memoryless network, illustrated in Fig. \ref{dm_mn}, is characterized by the conditional density $f_{Y_1,\ldots,Y_N|U_1,\ldots,U_N}$, which relates the real valued network inputs $(U_1,\ldots, U_N)$ to real valued network outputs $(Y_1,\ldots, Y_N)$.  
The set of source nodes is denoted as $\Sscr=\{s_1,s_2,\ldots,s_k\}\subseteq [1:N]$, and the set of destination nodes is denoted as $\Dscr=\{d_1,d_2,\ldots,d_k\}\subseteq [1:N]$.
The remaining nodes  (we assume without loss of generality that the sets of source and destination nodes have empty intersection) are relays $\mathcal{R}=\{r_1,r_2,\ldots,r_{N-2k}\}\subseteq [1:N]$.
Source node $s_m\in \Sscr$ has access to the i.i.d.\ source $X_m[t]$, $t=0,1,...$, which must be communicated to the corresponding destination node $d_m \in \Dscr$.
The i.i.d.\ vectors $(X_1[t],\ldots,X_k[t])$ have a joint distribution with covariance matrix $\vec K$.


\begin{definition} \label{codingdefn}
A coding scheme $\mathcal{C}$ with block length $n\in\N$ for distributed compression of a real valued memoryless source $(X_1,X_2,\ldots,X_k)$ over a $(k,N)$-memoryless network consists of the following:
\begin{enumerate}
\item {\emph{Source Encoding Functions}}: 
Source node $s_m\in \Sscr$ encodes the source $X_m$ as $U_{s_m}[t]=f_{s_m,t}(\vec X_m,Y_{s_m}^{t-1}),\ \forall\  t \in [0:n-1]$, where $f_{s_m,t}:{\R}^n\times \R^{t-1}\rightarrow \R$, $\forall\  m\in [1:k],\ \forall\  t \in [0:n-1]$ are the source encoding functions\footnote{Here and throughout we use the terms `functions' and `mappings' interchangeably and assume that they are measurable.}. 

\item {\textit{Relay Encoding Functions}}: Relay node $r_p\in \Rscr$ receives the channel outputs from the network and encodes it as $U_{r_p}[t]=f_{r_p,t}(Y_{r_p}^{t-1}),\ \forall\  t \in [0:n-1]$, where $f_{r_p,t}: \R^{t-1}\rightarrow \R$, $\forall\  p\in [1:N-2k],\ \forall\  t=[0:n-1]$, are the relay encoding functions. 

\item {\textit{Destination Encoding Functions}}: Destination node $d_m\in \Dscr$ receives the channel output from the network and encodes it as $U_{d_m}[t]=f_{d_m,t}(Y_{d_m}^{t-1})$, where $f_{d_m,t}: \R^{t-1}\rightarrow \R$, $\forall\  m\in [1:k],\ \forall\  t \in [0:n-1]$, are the destination encoding functions. 

\item {\textit{Destination Decoding Functions}}: At the end of the block of communication, each destination $d_m\in \Dscr$ constructs an estimate of the source as $\mathbf{\hat{X}}_m= g_{d_m}(\mathbf{Y}_{d_m})$, where $g_{d_m}: \R^{n}\rightarrow \R^n$, $\forall\  m\in [1:k]$, are the destination decoding functions. 
\end{enumerate}
\end{definition}

\begin{definition}
A distortion measure is a mapping $\dist : \R \times \R \to \R^+$.
\end{definition}

\begin{definition}
A distortion tuple $(D_1,D_2,\ldots,D_k)$ is said to be $\phi$-achievable if for some block length $n$, there exists a coding scheme $\mathcal{C}$, as described above, such that, 
\bea \label{distconst}
\frac{1}{n}  \E  \left[ \sum_{t=1}^n \dist(X_m[t],\hat X_m[t]) \right] \le D_m,\ \forall\ m \in [1:k]. 
\eea
\end{definition}


We focus on the quadratic distortion measure, i.e., where $\dist (x,y) = \qd(x,y) \defi (x-y)^2$.
Notice that, in this case, the expression in \eq{distconst} can be equivalently written as
\aln{
\frac{1}{n} \E\Big{[}\parallel \mathbf{X}_m-\hat{\mathbf{X}}_m\parallel^2\Big{]}\le D_m,\ \forall\ m \in [1:k]. 
}

\begin{definition}
The $\phi$-achievable distortion region $\D$ of a $(k,N)$-memoryless network is the closure of the set of achievable distortion tuples.
\end{definition}

\begin{theorem}[Main Result 1: \textbf{\textit{Worst-Case Source for $(k,N)$-memoryless network}}]
\label{theorem1}
For a $(k,N)$ memoryless network, let $\mathcal{D}^{source}_{NG}$ and $\mathcal{D}^{source}_{G}$ stand for the $\qd$-achievable distortion regions for an arbitrary memoryless non-Gaussian source with covariance matrix $\mathbf{K}$ and for a memoryless Gaussian source with the same covariance matrix, respectively. Then 
\bea
\mathcal{D}^{source}_{G} \subseteq \mathcal{D}^{source}_{NG}.
\eea
\end{theorem}

\begin{note}
A special case of Theorem \ref{theorem1} is that of wireline networks where each link is a (noiseless) bit pipe. This gives us the result of Gaussian source being the worst case source for the $k$-encoder distributed compression problem studied in \cite{wcsourceITW}.
\end{note}
In order to state our second main result, 
we focus on the following class of networks.

\begin{definition}
A $(k,N)$-memoryless network is said to be an additive-noise network if the input-output relationship is given by
\al{ \label{additiveio}
\begin{bmatrix}  Y_{1} \\  Y_{2} \\ \vdots \\  Y_{N} \end{bmatrix} & = 
H
\begin{bmatrix}  U_{1} \\  U_{2} \\ \vdots \\  U_{N} \end{bmatrix} + \begin{bmatrix}  Z_{1} \\  Z_{2} \\ \vdots \\  Z_{N} \end{bmatrix},
}
where $H$ is a real-valued $N \times N$ matrix and $(Z_1,\ldots,Z_N)$ is a noise vector with joint distribution $\mu_{\bf Z}$ independent of $(U_1,\ldots,U_N)$.
If $(Z_1,\ldots,Z_N)$ is distributed as $\Nscr({\bf 0},{\bf K})$ for some covariance matrix $\vec K$, then we call the network a $(k,N)$-additive white Gaussian noise (AWGN) network.
\end{definition}

\begin{theorem}[Main Result 2: \textbf{\textit{Worst-Case Noise for $(k,N)$-memoryless additive-noise network}}]
\label{theorem2}
For an arbitrary source of finite covariance and a $(k,N)$ memoryless additive-noise network, let $\mathcal{D}^{noise}_{NG}$ and $\mathcal{D}^{noise}_{G}$ stand for the $\qd$-achievable distortion regions for an arbitrary additive-noise non-Gaussian distribution with covariance matrix $\mathbf{K}$ and for additive Gaussian noise with the same covariance matrix, respectively. Then  
\bea
\mathcal{D}^{noise}_{G} \subseteq \mathcal{D}^{noise}_{NG}.
\eea
\end{theorem}

\begin{note}
In \cite{wcnoisefull}, Gaussian noise was also characterized as the worst-case additive noise in wireless networks.
However, \cite{wcnoisefull} considers a channel coding setting, and Gaussian noise is shown to minimize the capacity region, while, in this paper, we focus on a joint source-channel coding setting, and Theorem \ref{theorem2} establishes that Gaussian noises minimize the distortion region.
\end{note}

\section{Overview of Proof Ingredients}
\label{sec::preresults}

In this section, we give an overview of the main proof ingredients, describing at a high level how they are connected, and highlighting the connections between the proofs of Theorems \ref{theorem1} and \ref{theorem2}.

The overarching idea is to use a coding scheme for distributed compression designed for a Gaussian model (Gaussian source or Gaussian additive-noise network) to construct a new coding scheme that achieves approximately the same distortion tuple when the source or the additive noises are not Gaussian but have the same covariance as in the Gaussian case.

The first main step in the construction of this new coding scheme is to utilize the DFT-based linear transformation introduced in~\cite{wcnoisefull} in order to transform blocks of i.i.d. non-Gaussian random variables into ``approximately Gaussian'' random variables.
More specifically, we define the unitary $b \times b$ matrix $\vec Q$ (for simplicity we assume $b$ to be even) by setting the entry in the $(i+1)$th row and $(j+1)$th column to be
\al{ \label{qdef}
Q{(i,j)} =\left\{ \begin{array}{ll} 1/\sqrt{b} & \text{if $i = 0$} \\ \sqrt{2/b} \cos\left( \frac{2 \pi j i}{b} \right) & \text{if $i = 1,\ldots,\frac b2 - 1$} \\  (-1)^j/\sqrt{b} & \text{if $i = \frac b2 $} \\ \sqrt{2/b} \sin\left( \frac{2 \pi j (i-b/2) }{b} \right) & \text{if $i = \frac b2+1,\ldots,b - 1$} \end{array} \right. 
}
for $i,j \in \{0,\ldots,b-1\}$.
Applying $\vec Q$ to a vector $\vec x$ can be intuitively seen as first taking the $\DFT$ of $\vec x$, then separating the real and imaginary parts of the resulting vector, and renormalizing them so that the resulting transformation is unitary.
It is readily verified that $\vec Q$\footnote{Note that one can potentially come up with other choices for this transformation as well. Intuitively, $\mathbf{Q}$ should not put large mass on any of its components and distribute the mass almost uniformly. Mathematically, $\mathbf{Q}$ should be unitary and should satisfy the Lindenberg Condition, in the proof of Lemma 1, so as to have the corresponding Central Limit type theorem (Lemma 1). Our particular choice of $\mathbf{Q}$ is made for concreteness, mathematical convenience, and also due to the practical consideration that it would have a FFT-like implementation.} is a unitary transformation, i.e., that $\|\vec Q \vec x\| = \|\vec x\|$ for any $\vec x \in \R^b$.

The fact that, for any random vector $\vec x$ with i.i.d. non-Gaussian random variables, $\vec Q \vec x$ converges in distribution to a Gaussian random vector (as $b$ increases) was formalized in \cite{wcnoisefull,wcsourceITW}.
For random variables $X_1,X_2,...$ and $X$, we let $X_n \stackrel{d}{\to} X$ mean that $X_n$ converges in distribution to $X$ as $n \to \infty$.
%
The following lemma was first proven in \cite{wcsourceITW}, but we include a proof in Appendix \ref{appendixG} for completeness.

\begin{lemma}
[Convergence Lemma] 
\label{convlemma}
Suppose $\left\{ \left( X_1[i],\ldots,X_k[i]\right) \right\}_{i=0}^{nb-1}$ is an i.i.d.~sequence of length-$k$ random vectors with covariance matrix $\bf K$, and let $\mathbf{Q}$ be the unitary linear transformation in (\ref{qdef}) and
\al{ \label{tildex}
& \begin{bmatrix} \tilde X_{1}^{(0)}[t] & \cdots & \tilde X_{k}^{(0)}[t] \\ 
\tilde X_{1}^{(1)}[t] & \cdots & \tilde X_{k}^{(1)}[t]  \\ \vdots & \ddots & \vdots \\ 
\tilde X_{1}^{(b-1)}[t] & \cdots & \tilde X_{k}^{(b-1)}[t] \end{bmatrix}  =
 {\bf Q} 
\begin{bmatrix} X_1[t b] & \cdots & X_k[t b] \\ 
X_1[t b+1] & \cdots & X_k[t b+1]  \\ \vdots & \ddots & \vdots \\ 
X_1[t b+b-1] & \cdots & X_k[t b+b-1] \end{bmatrix}
}
for $t = 0,1,\ldots,n-1$.
Then, for any sequence $\ell_b$ such that, for $b=1,2,...$, $\ell_b \in \{0,1,\ldots,b-1\}$, and any $t \in \{0,1,\ldots,n-1\}$, 
\aln{
\left( \tilde X_1^{(\ell_b)}[t],\ldots,\til X_k^{(\ell_b)}[t]\right) \stackrel{d}{\to} \Nscr({\bf 0},{\bf K}), \text{ as $b \to \infty$}.}
\end{lemma}

In the proof of Theorem \ref{theorem1}, we apply $\vec Q$ to blocks of $b$ source symbols in order to create an effective source which is approximately Gaussian.
Similarly, in the proof of Theorem \ref{theorem2}, we apply $\vec Q^{-1}$ to blocks of $b$ transmit signals, and $\vec Q$ to blocks of $b$ received signals, in order to make the effective additive noises approximately Gaussian.
Since the application of the linear tranformation $\vec Q$ (and $\vec Q^{-1}$) results in statistical dependencies between the resulting sources or noises, a simple interleaving scheme is employed in order to create i.i.d. approximately Gaussian sources or noises.
We then apply a coding scheme designed to achieve a given distortion tuple $(D_1,\ldots,D_k)$ under Gaussian assumptions to these resulting i.i.d. blocks.


The main technical challenge in the proofs of Theorems \ref{theorem1} and \ref{theorem2} is to show that, as $b \to \infty$, the resulting distortion converges to $(D_1,\ldots,D_k)$.
In order to do this, we establish several technical lemmas, which allow us to assume without loss of generality that our original coding scheme designed for a Gaussian model satisfies certain properties.


First, 
we need a lemma that allows us to restrict attention to \emph{bounded output} coding schemes; i.e., coding schemes in which the output of the decoding functions is bounded.
Thus we need to show that any achievable distortion tuple can be attained arbitrarily closely by a bounded output scheme. 
The advantage of dealing with codings schemes with bounded output is that it becomes easier to apply
standard results such as the Dominated Convergence Theorem to the associated sequence of distortions. 

\begin{lemma} [Bounded Output Lemma]
\label{boundedoutputlemma}
Suppose $(X_1[t],\ldots,X_k[t])$ has an arbitrary joint distribution with covariance matrix $\bv K$ and a coding scheme $\C$ with blocklength $n$ achieves distortion vector $(D_1,\ldots,D_k)$.
Then, for any $\ep > 0$, one can build another coding scheme $\tilde \C$ of block length $n$ with decoding functions $\tilde g_{d_m}$ with the property that 
\aln{
\left\| \tilde g_{d_j} ( y_1,\ldots,y_n) \right\|_\infty \leq M,
}
for any $(y_1,\ldots,y_n) \in \R^n$, $j = 1,\ldots,k$ and a fixed $M > 0$, which achieves distortion vector $(D_1+\ep,\ldots,D_k+\ep)$.
\end{lemma}


Another important property that we need to assume for the original coding scheme designed for a Gaussian model is that of \emph{finite reading precision}, which was introduced in \cite{wcnoisefull}.
In coding schemes with finite precision, the encoding and decoding operations of the nodes in the network may only take finite precision inputs; i.e., inputs with a finite number of decimal places.
More formally, for a real-valued vector $x^n = (x_1,\ldots,x_n)$ and a positive integer $\rho$, we let $\left\lfloor x^n \right\rfloor_{\rho} = 2^{-\rho}\left(\lfloor 2^\rho x_1 \rfloor,\ldots,\lfloor 2^\rho x_n \rfloor\right)$, and define the following.

\begin{definition} \label{finiteprecisiondefn2}
A coding scheme $\mathcal{C}$ of block length $n$ is said to have
\textit{finite reading precision} $\rho=[\rho_1,\ldots,\rho_N]\in\N^N$ if the encoding function at each source $s_m \in \mathcal{S}$ satisfiess
\aln{
f_{s_m,t} ( x_m^n, y^{t-1}) = f_{s_m,t} (  x_m^n, \left\lfloor y^{t-1}\right\rfloor_{\rho_{s_m}}),
}
and the encoding functions at each node $i \in \Rscr \cup \Dscr$ satisfies
\aln{
f_{i,t} (y^{t-1}) = f_{i,t} ( \left\lfloor y^{t-1}\right\rfloor_{\rho_i}),
}
for any $x_m^n \in \R^n$, any $y^{t-1} \in \R^{t-1}$, and any time $t$.
\end{definition}

While finite reading precision is useful in the proof of Theorem \ref{theorem2}, to prove Theorem \ref{theorem1}, we instead require that the source nodes only have access to a finite number of decimal places of the source symbols.
We call this \emph{finite encoding precision}.

\begin{definition} \label{finiteprecisiondefn}
A coding scheme $\mathcal{C}$ with block length $n$ is said to have
\textit{finite encoding precision} $\rho=[\rho_1,\ldots,\rho_k]\in\N^k$ if the encoding function at each source $s_m \in \mathcal{S}$ satisfies
\aln{
f_{s_m,t} ( x_m^n, y^{t-1}) = f_{s_m,t} ( \left\lfloor x_m^n \right\rfloor_{\rho_m}, y^{t-1}),\ \forall\ m\in[1:k]
}
for any $x_m^n \in \R^n$, any $y^{t-1} \in \R^{t-1}$, and any time $t$.
\end{definition}




In order to prove the optimality of coding schemes with finite reading/encoding precision (i.e., that they can come arbitrarily close to achieving any point in the achievable distortion region), our main tool is the following result.


\begin{lemma} \label{densitylemma}
Suppose $\mathbf{Y}=(Y_1,\ldots,Y_i,\ldots,Y_k)$ is a random vector with density $f_{Y_1,\ldots,Y_i,\ldots,Y_k}$. 
Consider some $\rho\in\N$. 
For some $i\in[1:k]$, let $\tilde{Y}_{i}^{(\rho)}=\lfloor Y_{i}\rfloor_{\rho}+U_{\rho}$,  where $U_{\rho}$ is uniformly distributed in $(-2^{-\rho-1},2^{-\rho-1})$ and is independent of $\mathbf{Y}$. 
Then
\bea
\lim_{\rho\rightarrow\infty}f_{Y_1,\ldots,\tilde{Y}_i^{(\rho)},\ldots,Y_k}(y_1,\ldots,y_i,\ldots,y_k)=f_{Y_1,\ldots,Y_i,\ldots,Y_k}(y_1,\ldots,y_i,\ldots,y_k), \ \forall\ i\in[1:k],
\eea
for almost every $(y_1,\ldots,y_i,\ldots,y_k) \in \R^k$.
\end{lemma}

This lemma allows us to take a coding scheme that does not have finite precision, and consider finer and finer discretizations of its encoding functions, in a way that the resulting distortion tuple approaches that of the original coding scheme.

Another technical tool that is useful in proving the optimality of coding schemes with finite encoding precision is the following lemma.
Intuitively, it allows us to view our stochastic network (as defined in Section \ref{sec::problem}) as a collection of deterministic networks, which facilitates the bounding of the resulting distortion.

\begin{lemma}[Functional Representation Lemma] \label{frlemma}
For any two random vectors $Y$ and $U$, there exist a (deterministic, measurable) function $h$ and a random vector $Q$, independent of $U$, for which the pair $(h(U,Q),U)$ has the same distribution as $(Y,U)$.
\end{lemma}

Lemmas \ref{boundedoutputlemma}, \ref{densitylemma} and \ref{frlemma} allow us to prove the essential optimality of finite precision coding schemes, stated in the next two lemmas, whose proofs are presented in Appendices \ref{appendixD} and \ref{appendixE}.

\begin{lemma}[Finite Encoding Precision Lemma]\label{finiteprecisionlemma}
Suppose the distortion tuple $(D_1,\ldots,D_k)$ is achievable over the $(k,N)$-memoryless
network. 
Then, for any $\epsilon>0$, there exists a coding scheme with finite encoding precision that achieves distortion tuple
$(D_1+\epsilon,\ldots,D_k+\epsilon)$.
\end{lemma}

%
%
%

\begin{lemma}[Finite Reading Precision Lemma]\label{finiteprecisionlemma2}
Suppose the distortion tuple $(D_1,\ldots,D_k)$ is achievable over the $(k,N)$-AWGN network. 
Then, for any $\epsilon>0$, there exists a coding scheme with finite reading precision that achieves distortion tuple
$(D_1+\epsilon,\ldots,D_k+\epsilon)$.
\end{lemma}

\noindent {\bf Remark.}\;We point out that, from the proofs of Lemmas \ref{boundedoutputlemma}, \ref{finiteprecisionlemma} and \ref{finiteprecisionlemma2}, it can be seen that there exists a \emph{single} coding scheme that has both bounded outputs and finite encoding/reading precision and achieves distortion tuple $(D_1+\epsilon,\ldots,D_k+\epsilon)$.

\vspace{3mm}

In order to state the next result, we use the following definition.

\begin{definition}
A function $f : \R^a \to \R^b$ is locally constant at a point $x \in \R^a$ if it is constant in some neighborhood of $x$.
\end{definition}

The importance of finite encoding/reading precision is expressed in the following lemma. 

\begin{lemma}[Continuity Lemma] \label{continuouslemma}
If a function $f : \R^a \to \R^b$ satisfies
\aln{
f( \bv x) = f( \left\lfloor \bv x \right\rfloor_{\rho})
}
for some $\rho \in \N$ and any $\bv x \in \R^a$, then $f$ is locally constant (and, thus, continuous) almost everywhere.
\end{lemma}

Therefore, coding schemes with finite precision have encoding functions that are continuous almost everywhere.
As a result, we may start off the proofs of Theorems \ref{theorem1} and \ref{theorem2} with a coding scheme (designed for a Gaussian model) that has bounded outputs and finite encoding precision (in the case of Theorem \ref{theorem1}) or finite reading precision (in the case of Theorem \ref{theorem2}).
The continuity of these functions allows us to bound the distortion achieved by the coding scheme constructed through the application of the linear transformation $\vec Q$ and the interleaving scheme, and show that it converges to the distortion of the original coding scheme as $b \to \infty$,  by invoking the results pertaining to weak convergence to the Gaussian distribution in the transform domain.

\section{Worst Case Source for a given Network}
\label{sec::worstsource}
%

\begin{figure}[t]
        \centering
        \subfigure[Overview of the proof of Theorem \ref{theorem1}.]{
\scalebox{0.44}{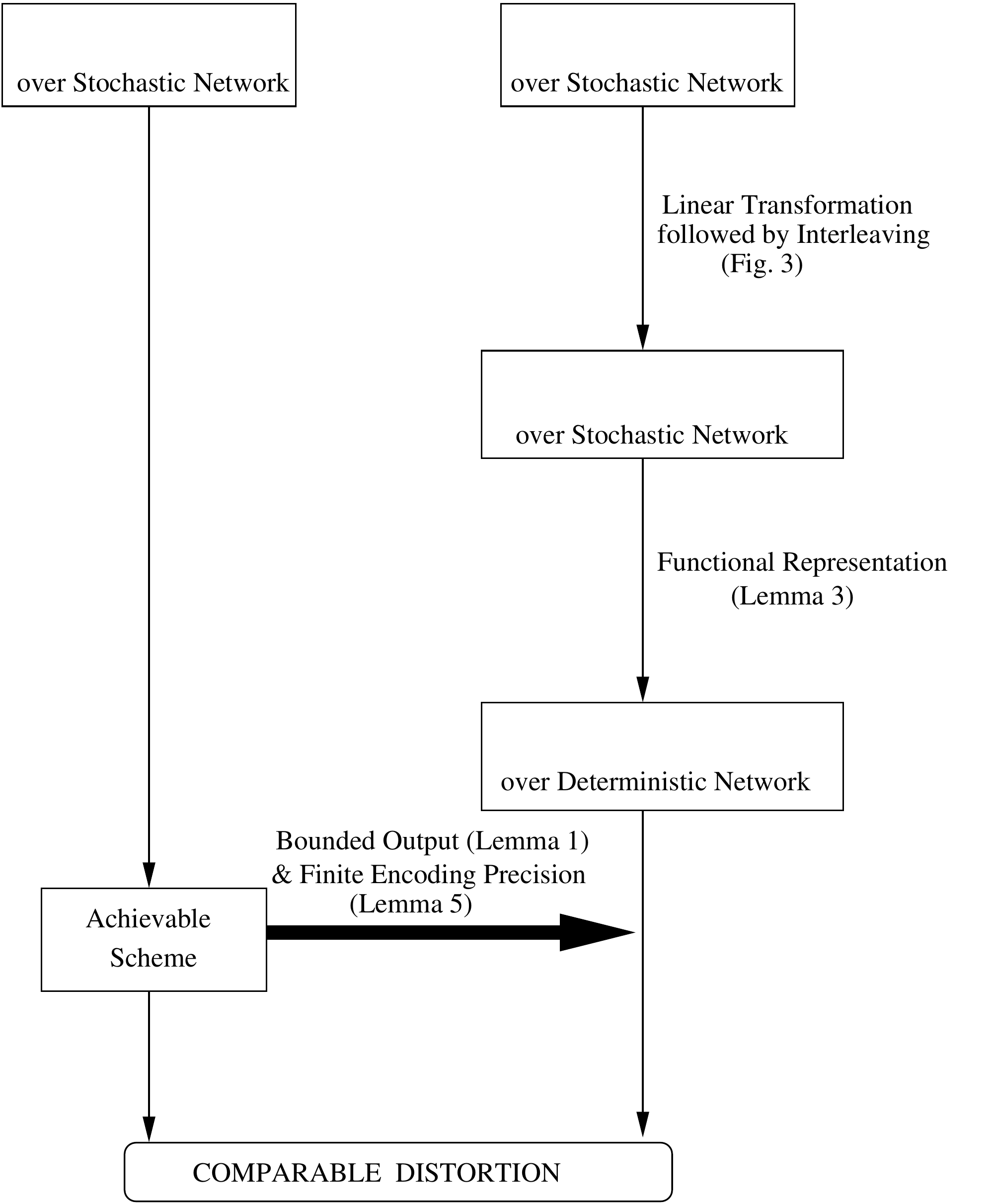}

\label{fig::worstsource}
}
\subfigure[Overview of the proof of Theorem \ref{theorem2}.]{
\scalebox{0.44}{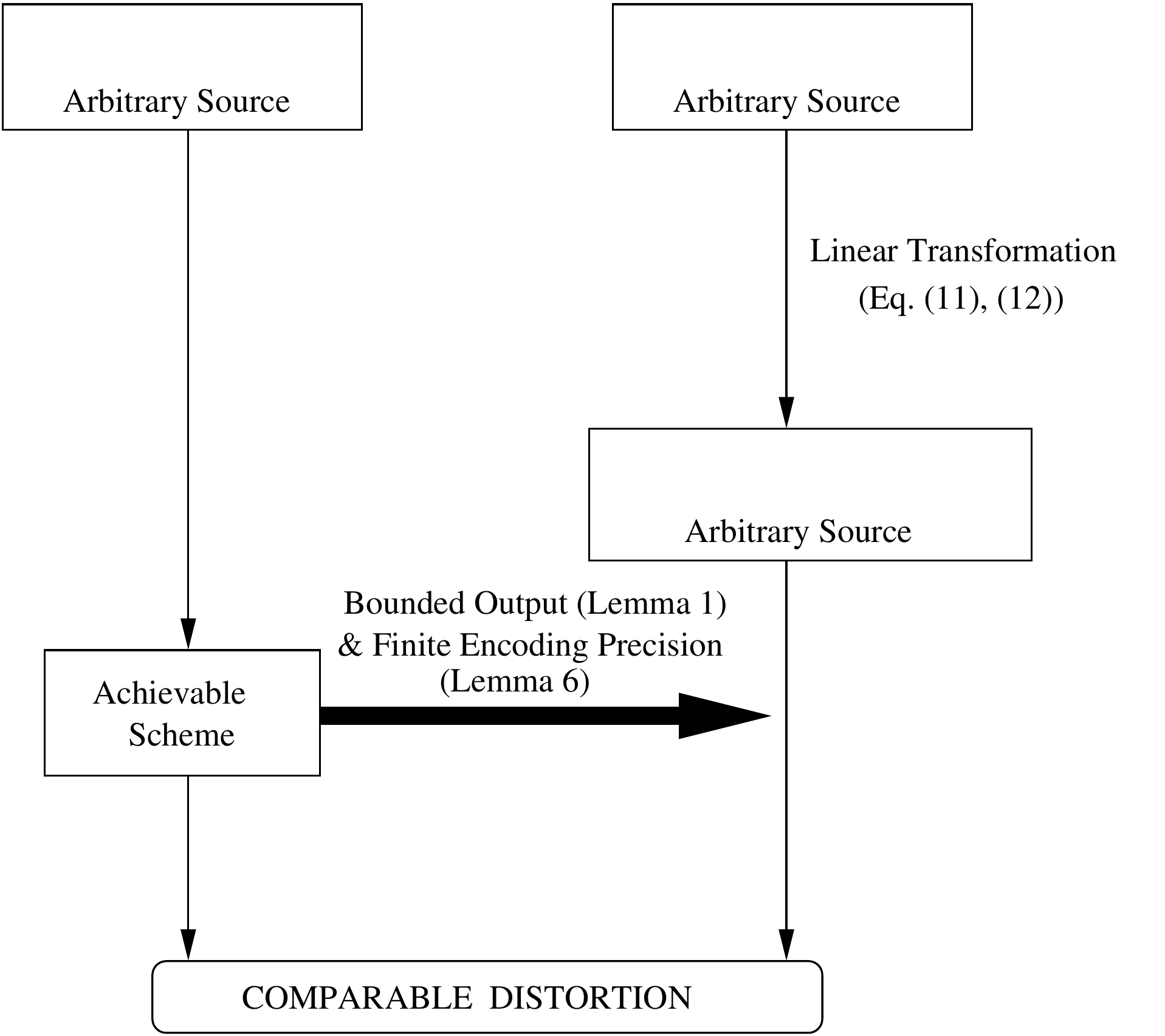}
\label{fig::worstnoise}
}
        \caption{Flow Diagrams showing the overview of proofs of Theorems \ref{theorem1} and \ref{theorem2}}
        \label{fig:overview}
\end{figure}

%

In this section, we prove Theorem \ref{theorem1}; that is, that among all the sources with a given covariance, the Gaussian source is least compressible over a fixed memoryless network.
 Note that proving Theorem \ref{theorem1} is equivalent to proving the following theorem.
\begin{theorem}[Equivalent to Theorem \ref{theorem1}] \label{thm1b}
If a distortion tuple $(D_1,\ldots,D_k)$ is $\qd$-achievable when $({X}_1,\ldots,{X}_k) $ is jointly Gaussian with covariance matrix ${\mathbf{K}}$, then for any $\epsilon>0$, the distortion tuple $(D_1+\epsilon,\ldots,D_k+\epsilon)$ is $\xi$-achievable when  $({X}_1,\ldots,{X}_k) $ has an arbitrary distribution with covariance matrix ${\mathbf{K}}$.
\end{theorem}

Before delving into the mathematical details of  the  proof below, we give an overview of the proof steps. 
A high level illustration is provided in Fig. \ref{fig::worstsource}. 
The overall idea is to use the achievable scheme for distributed compression of the  Gaussian source over the network, and devise a new scheme for the arbitrary source with the same covariance to show that the achievable distortion in both cases is similar.

We first apply the linear transformation $\vec Q$ to blocks of the sources, and then the transformed source symbols from various blocks are interleaved to create new blocks of independent ``effective'' source symbols. 
The idea is that each of these new effective source symbols is now a weighted aggregate of many source symbols, and using a central limit theorem-like result (Lemma \ref{convlemma}), the effective source is close to the Gaussian source with the same covariance. 
We next invoke the Functional Representation Lemma (Lemma \ref{frlemma}) to ``transform'' the given stochastic network into a randomly chosen deterministic network; that is to say that the output to a node is a deterministic function of the inputs to the network and some external  randomness, independent of the inputs. 
We then take the achievable scheme for the Gaussian source, and construct an equivalent scheme with bounded output (using Lemma \ref{boundedoutputlemma}) and finite encoding precision (using Lemma \ref{finiteprecisionlemma}), and apply it to the new effective network with approximately Gaussian sources. 
Using the continuity property of finite encoding precision schemes (Lemma \ref{continuouslemma}), and the property of bounded output, we conclude the proof by showing that the distortion achieved on this effective network with approximately Gaussian sources is close to what it would have been if the sources were actually Gaussian. 

\vspace{5mm}

\begin{proof}[Proof of Theorem \ref{thm1b}]
Suppose the distortion tuple $(D_1,\ldots,D_k)$ is achievable in the case where $(X_1[0],\ldots,X_k[0])$ is jointly Gaussian with covariance matrix $\bv K$.
Fix $\ep > 0$.
From Lemmas \ref{boundedoutputlemma} and \ref{finiteprecisionlemma} (and the subsequent remark), we can assume that we have a code $\C$ with blocklength $n$ as defined in Definition \ref{codingdefn}, which 
achieves distortion vector $(D_1 +\ep/2,\ldots,D_k+\ep/2)$
if $(X_1[0],\ldots,X_k[0])$ is jointly Gaussian, with finite encoding precision $\rho=[\rho_1,\ldots,\rho_k]\in\N^k$ and for which
\aln{
\left\| g_{d_j} ( y_1,\ldots,y_n) \right\|_\infty \leq M,
}
for any $(y_1,\ldots,y_n) \in \R^n$, $j = 1,\ldots,k$ and a fixed $M > 0$.

We will build a coding scheme $\tilde \C$ with block length $nb$, for a large integer $b$, with source encoding functions $\tilde f_{s_m,t}$, relay encoding functions $\tilde f_{r_p,t}$, destination encoding functions $\tilde f_{d_m,t}$ and destination decoding functions $\tilde g_{d_m,t}$.
Since we will be working with a block length $nb$, we will let $\bv X_m = (X_m[0],\ldots,X_m[nb-1])$, for $m=1,\ldots,k$.
All relay encoding functions and destination encoding functions will be constructed by simply repeating $f_{r_p,t}$ and $f_{d_m,t}$, $b$ times.
More precisely, for a time $t = \ell n + \tau$, for $\ell \in \{0,\ldots,b-1\}$ and $\tau \in \{0,\ldots,n-1\}$, we let
\aln{
& \tilde f_{r_p,t}(y^{t-1}) = f_{r_p,\tau}( y[\ell n],y[\ell n + 1],\ldots,y[\ell n + \tau-1]), \\
& \tilde f_{d_m,t}(y^{t-1}) = f_{d_m,\tau}( y[\ell n],y[\ell n + 1],\ldots,y[\ell n + \tau-1]).
}
Thus, from the point of view of the relays and destination encoding functions, we are simply repeating coding scheme $\C$, $b$ independent times.

As in the construction of the relay and destination encoding functions, we will essentially repeat $f_{s_m}$ $b$ times.
However, instead of applying each of these source encoding functions considering the random source $\vec X_m$, each source $s_m$ will instead apply it to an effective random source $\vec {\tilde X}_m$ which can be obtained from $\vec X_m$ through an invertible transformation.
This invertible transformation is depicted in Fig. \ref{encodingfig}.
\begin{figure}[ht] 
     \centering
       \includegraphics[width=0.5\linewidth]{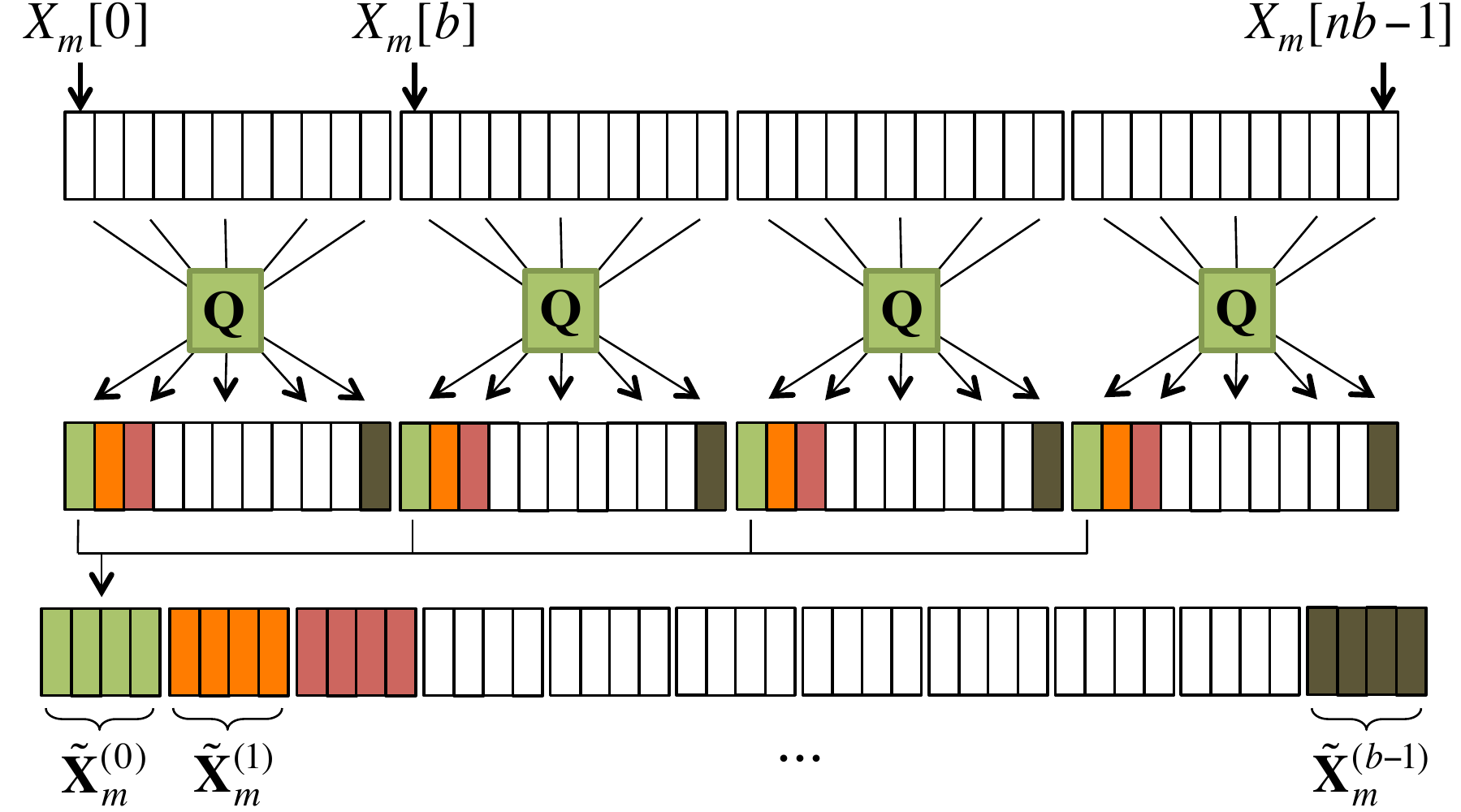} \caption{Illustration of the new encoding procedure for source node $s_m$. \label{encodingfig}}
\end{figure}
First, $\vec X_m$ is broken into $n$ blocks of length $b$.
To each of these blocks we apply the unitary linear transformation $\vec Q$.
The $n$ resulting blocks of length $b$ are then interleaved, generating $b$ length-$n$ vectors $\til{\vec X}_m^{(0)},\ldots,\til{\vec X}_m^{(b-1)}$, as shown in Figure \ref{encodingfig}.
The new effective source is given by $\til{\vec X}_m = (\til{\vec X}_m^{(0)},\ldots,\til{\vec X}_m^{(b-1)})$.
Therefore, for a time $t = \ell n + \tau$, for $\ell \in \{0,\ldots,b-1\}$ and $\tau \in \{0,\ldots,n-1\}$, the source encoding function can be described as
\aln{
& \tilde f_{s_m,t}(\vec X_m,y^{t-1}) = f_{s_m,\tau}(\vec {\tilde X}_m^{(\ell)}, y[\ell n],y[\ell n + 1],\ldots,y[\ell n + \tau-1]).
}
%
At the decoders side, we first apply the destination decoding functions $g_{d_m}$ to each block of length $n$, and then we invert the transformation applied by source $s_m$ to $\vec {\tilde X}_m$.
More precisely, $\tilde g_{d_m}$ is obtained by taking the length-$nb$ vector
\aln{
\left( g_{d_m}( Y_{d_m}[0],\ldots,Y_{d_m}[n-1]), g_{d_m}( Y_{d_m}[n],\ldots,Y_{d_m}[2n-1]),\ldots, g_{d_m}( Y_{d_m}[(b-1)n],\ldots,Y_{d_m}[bn-1]) \right),
}
interleaving the $b$ blocks of length $n$ to obtain $n$ blocks of length $b$, and then applying $\vec Q^{-1}$ to each of these blocks.


Our next goal is to show that, by choosing $b$ large enough, we can make the distortion of this new code arbitrarily close to the distortion of the original code applied to the Gaussian source.
We will employ the Functional Representation Lemma (Lemma \ref{frlemma}) in order to alternatively view our stochastic network as an ensemble of deterministic networks. 
Note that the achievable distortion of the network under consideration is a function of the joint distribution of sources, channel inputs, channel outputs, i.e., $(X_1,\ldots,X_k,Y_1,\ldots,Y_N,U_1,\ldots,U_N)$. 
Hence, from Lemma \ref{frlemma}, the memoryless network $f_{Y_1,\ldots,Y_N|U_1,\ldots,U_N}$ can be equivalently
represented as a deterministic network 
 $Y_i=h_i(U_1,\ldots,U_N,Z),\ \forall \
i \in [1:N]$ where $Z$ is a random vector independent of the channel inputs $(U_1,\ldots,U_N)$.
Thus, given a length-$nb$ vector of realizations of $Z$, $\vec z = (z_0,\ldots,z_{nb-1})$, given our coding scheme $\tilde C$, all the received signals in the network are a deterministic function of the random sources $\vec X_1, \vec X_2,\ldots,\vec X_k$.
Thus for a block length $nb$, we can write, for some functions $F_i$,
$\mathbf{Y}_i=F_i(\mathbf{X}_{1},\mathbf{X}_{2},\ldots,\mathbf{X}_{k},
\vec Z),\ \forall\ i\in [1:N]$.
Therefore, once we condition on $\vec Z = \vec z$ for some realization $\vec z$ of $\vec Z$, the distortion of the reconstruction of destination $d_m$,
\aln{
\frac1{nb} \left\| \vec{X}_m  - \tilde g_{d_m}\left( F_{d_m}(\vec X_1, \vec X_2,\ldots,\vec X_k,\vec z)\right)  \right\|^2,
}
is only a function of the random sources $\vec X_1, \vec X_2,\ldots,\vec X_k$.

Next, we notice that, since $\vec Q$ is a unitary linear transformation, the distortion of $\tilde \C$ can be written in terms of ${\vec {\til X}}_m^{(\ell)}$ for $\ell = 0,\ldots,b-1$ as
\aln{
\frac1{b} \sum_{\ell = 0}^{b-1} \frac1n\left\| \vec{\til X}_m^{(\ell)} - g_{d_m}\left( \tilde F_{d_m}(\vec{\til X}_1^{(\ell)},\ldots,\vec{\til X}_k^{(\ell)},\vec z) \right)  \right\|^2,
}
where $\tilde F_{d_m}$ are defined again through the Functional Representation Lemma, by noticing that, during times $t=\ell n, \ell n + 1,\ldots, (\ell+1) n -1$, the received signals at $d_m$ only depend on $(\vec{\til X}_1^{(\ell)},\ldots,\vec{\til X}_k^{(\ell)})$, and not the entire sources $\vec X_1, \vec X_2,\ldots,\vec X_k$.
For each $b=1,2,...$, we will let 
\aln{
\ell_b = \arg \max_{0 \leq \ell \leq b-1} \E \left\| \vec{\til X}_m^{(\ell)} - g_{d_m}\left( \tilde F_{d_m}(\vec{\til X}_1^{(\ell)},\ldots,\vec{\til X}_k^{(\ell)},\vec z) \right)  \right\|^2,
}
i.e., the $\ell_b$th length-$n$ block has the largest expected distortion.
Note that $\left\{ \left(\til X_1^{(\ell_b)}[i],\ldots,\til X_k^{(\ell_b)}[i]\right) \right\}_{i=0}^{n-1}$ is an i.i.d.~sequence of length-$k$ random vectors.
From Lemma \ref{convlemma}, we see that it converges in distribution to a sequence of i.i.d.~jointly Gaussian random vectors with covariance matrix $\vec K$, as $b \to \infty$.

Now, from Lemma \ref{continuouslemma}, each of the source encoding functions $f_{s_m,t}$ of the original coding scheme $\C$ are constant almost everywhere, since they have finite encoding precision.
Since each function $\tilde F_i$ can only depend on the effective sources $\vec{\til X}_1^{(\ell_b)},\ldots,\vec{\til X}_k^{(\ell_b)}$ through the source encoding functions $f_{s_m,t}$, it is not difficult to see that, for a fixed $\vec z$, $\tilde F_i \left( \vec{\til X}_1^{(\ell_b)},\ldots,\vec{\til X}_k^{(\ell_b)}, \vec z\right)$ is an almost-everywhere-constant function of  $\vec{\til X}_1^{(\ell_b)},\ldots,\vec{\til X}_k^{(\ell_b)}$.
Therefore, the mapping
\aln{
\left\{\vec{\til X}_m^{(\ell_b)} \right\}_{m=1}^k \mapsto \left\| \vec{\til X}_m^{(\ell_b)} - g_{d_m}\left({\tilde F}_{d_m}(\vec{\til X}_1^{(\ell_b)},\ldots,\vec{\til X}_k^{(\ell_b)},\vec z) \right)  \right\|^2,
}
for $m=1,\ldots,k$, is continuous almost everywhere.
We conclude that 
\aln{ 
& \left\| \vec{\til X}_m^{(\ell_b)} - g_{d_m}\left( {\tilde F}_{d_m}(\vec{\til X}_1^{(\ell_b)},\ldots,\vec{\til X}_k^{(\ell_b)},\vec z) \right)  \right\|^2 
 \stackrel{d}{\to} 
\left\| \vec{X}^G_m - g_{d_m}\left( {\tilde F}_{d_m}(\vec{X}^G_1,\ldots,\vec {X}^G_k, \vec z) \right)  \right\|^2,
}
as $b \to \infty$, 
where ${\vec X}^G_m = (X^G_m[0],\ldots,X^G_m[n-1])$, for $m=1,\ldots,k$, and $\left\{ \left(X^G_1[i],\ldots,X^G_k[i]\right) \right\}_{i=0}^{n-1}$  is an i.i.d.~sequence such that $(X^G_1[0],\ldots,X^G_k[0])$ is jointly Gaussian with zero mean and covariance matrix $\vec K$.
Moreover, we have that 
\al{ 
\left\| \vec{\til X}_m^{(\ell_b)} - g_{d_m}\left( {\tilde F}_{d_m}(\vec{\til X}_1^{(\ell_b)},\ldots,\vec{\til X}_k^{(\ell_b)},\vec z) \right)  \right\|^2  & \leq 2 \left\| \vec{\til X}_m^{(\ell_b)} \right\|^2 + 2 \left\| g_{d_m}\left( {\tilde F}_{d_m}(\vec{\til X}_1^{(\ell_b)},\ldots,\vec{\til X}_k^{(\ell_b)},\vec z) \right)  \right\|^2 \nonumber \\
& \leq 2 \left\| \vec{\til X}_m^{(\ell_b)} \right\|^2 + 2 n M^2, \label{bound1}
}
and also that
\al{ 
& \E \left\|\vec{\til X}_m^{(\ell_b)} \right\|^2 = n \E \left(  \sum_{j=0}^{b-1} X_m[j] ~ Q(\ell_b,j) \right)^2 = n \vec K_{m,m} \sum_{j=0}^{b-1} Q^2(\ell_b,j)  = n \vec K_{m,m} < \infty. \label{bound2}
}
Thus, from a variation of the Dominated Convergence Theorem (see Appendix \ref{appendixH}), we conclude that, as $b\to\infty$,
\al{
\E\left[ \left. \left\| \vec{\til X}_m^{(\ell_b)} - g_{d_m}\left( {\tilde F}_{d_m}(\vec{\til X}_1^{(\ell_b)},\ldots,\vec{\til X}_k^{(\ell_b)},\vec Z) \right)  \right\|^2 \right| \vec Z = \vec z \right] 
  \to \E\left[ \left. \left\| \vec{X}^G_m - g_{d_m}\left( {\tilde F}_{d_m}(\vec{X}^G_1,\ldots,\vec {X}^G_k, \vec Z) \right)  \right\|^2 \right| \vec Z = \vec z \right], 
\label{condexpz}
}
for all $\vec z$.
This means that the random variable $\E\left[ \left. \left\| \vec{\til X}_m^{(\ell_b)} - g_{d_m}\left( {\tilde F}_{d_m}(\vec{\til X}_1^{(\ell_b)},\ldots,\vec{\til X}_k^{(\ell_b)},\vec Z) \right)  \right\|^2 \right| \vec Z \right]$ converges surely to $\E\left[ \left. \left\| \vec{X}^G_m - g_{d_m}\left( {\tilde F}_{d_m}(\vec{X}^G_1,\ldots,\vec {X}^G_k, \vec Z) \right)  \right\|^2 \right| \vec Z \right]$.
Moreover, for any $b$, by following the steps in \eq{bound1} and \eq{bound2}, since $\vec{\til X}_m^{(\ell_b)}$ is independent of $\vec Z$, 
\aln{
\E\left[ \left. \left\| \vec{\til X}_m^{(\ell_b)} - g_{d_m}\left( {\tilde F}_{d_m}(\vec{\til X}_1^{(\ell_b)},\ldots,\vec{\til X}_k^{(\ell_b)},\vec Z) \right)  \right\|^2 \right| \vec Z = \vec z \right] &  \leq 2 \ \E \left[ \left. \left\|\vec{\til X}_m^{(\ell_b)} \right\|^2 \right| \vec Z = \vec z \right] + 2 n M^2 \\
&  = 2 \ \E \left[  \left\|\vec{\til X}_m^{(\ell_b)} \right\|^2  \right] + 2 n M^2 \\
& \leq 2 n \left( \vec K_{m,m} + M^2 \right).
}
Thus, a second application of the Dominated Convergence Theorem yields
\aln{
& \E\left[ \E\left[ \left. \left\| \vec{\til X}_m^{(\ell_b)} - g_{d_m}\left( {\tilde F}_{d_m}(\vec{\til X}_1^{(\ell_b)},\ldots,\vec{\til X}_k^{(\ell_b)},\vec Z) \right)  \right\|^2 \right| \vec Z \right] \right] \\
& \quad \to \E\left[ \E\left[ \left. \left\| \vec{X}^G_m - g_{d_m}\left( {\tilde F}_{d_m}(\vec{X}^G_1,\ldots,\vec {X}^G_k, \vec Z) \right)  \right\|^2 \right| \vec Z \right] \right],
}
which implies
\aln{
& \E\left[  \left\| \vec{\til X}_m^{(\ell_b)} - g_{d_m}\left( {\tilde F}_{d_m}(\vec{\til X}_1^{(\ell_b)},\ldots,\vec{\til X}_k^{(\ell_b)},\vec Z) \right)  \right\|^2  \right]  \to \E\left[ \left\| \vec{X}^G_m - g_{d_m}\left( {\tilde F}_{d_m}(\vec{X}^G_1,\ldots,\vec {X}^G_k, \vec Z) \right)  \right\|^2\right] \leq n(D_m + \ep/2).
}
Therefore, we can choose $b$ sufficiently large so that 
\aln{
& \frac1n \E\left[  \left\| \vec{\til X}_m^{(\ell_b)} - g_{d_m}\left( F_{d_m}(\vec{\til X}_1^{(\ell_b)},\ldots,\vec{\til X}_k^{(\ell_b)},\vec Z) \right)  \right\|^2  \right] \\ 
& \quad \quad \quad \quad \leq \frac1n \E\left[ \left\| \vec{X}^G_m - g_{d_m}\left( {\tilde F}_{d_m}(\vec{X}^G_1,\ldots,\vec {X}^G_k, \vec Z) \right)  \right\|^2\right] + \ep/2 \\
& \quad \quad \quad \quad \leq D_m + \ep.
}
The expected distortion of code $\tilde \C$ (with blocklength $nb$) thus satisfies
\aln{
& \frac1{nb} \sum_{\ell = 0}^{b-1} \E\left[  \left\| \vec{\til X}_m^{(\ell)} - g_{d_m}\left( {\tilde F}_{d_m}(\vec{\til X}_1^{(\ell)},\ldots,\vec{\til X}_k^{(\ell)},\vec Z) \right)  \right\|^2  \right] \\ 
& \quad \quad \quad \quad \leq \frac1n \E\left[  \left\| \vec{\til X}_m^{(\ell_b)} - g_{d_m}\left( {\tilde F}_{d_m}(\vec{\til X}_1^{(\ell_b)},\ldots,\vec{\til X}_k^{(\ell_b)},\vec Z) \right)  \right\|^2  \right] \\
& \quad \quad \quad \quad \leq D_m + \ep,
}
for $m=1,\ldots,k$.
This concludes the proof of Theorem \ref{theorem1}. 
\end{proof}

\section{Worst Case Additive Noise for a given Source}
\label{sec::worstnoise}

In this section, we prove Theorem \ref{theorem2}; that is, that given any arbitrary source to be encoded over an additive-noise network, amongst all the noise distributions with a given covariance, Gaussian noise leads to the worst distortion. Note that proving Theorem \ref{theorem2} is equivalent to proving the following theorem.
\begin{theorem}[Equivalent to Theorem \ref{theorem2}] \label{thm2b}
If a distortion tuple $(D_1,\ldots,D_k)$ is $\xi$-achievable on an additive-noise memoryless network when $(Z_1,\ldots,Z_N)$ is jointly Gaussian with covariance matrix ${\mathbf{K}}$, then for any $\epsilon>0$, the distortion tuple $(D_1+\epsilon,\ldots,D_k+\epsilon)$ is $\qd$-achievable when $(Z_1,\ldots,Z_N)$ has an arbitrary distribution with covariance matrix ${\mathbf{K}}$.
\end{theorem}

As we did in the previous section, we first give an overview of the proof. A high level illustration is provided in Fig. \ref{fig::worstnoise}. The overall idea is to use the achievable scheme for the distributed compression of the source over the additive Gaussian noise network and devise a new scheme when the noise is arbitrarily distributed with the same covariance so that the achievable distortion in both cases is comparable. 

We first apply the linear transformation $\vec Q$ to the blocks of the network inputs and outputs to create effective inputs and outputs. Our goal is to convert the additive-noise network into an approximately additive Gaussian noise network. The key idea is that the new effective channel noise is now a weighted aggregate of many independent noise components, and using a central limit theorem like result (Lemma \ref{convlemma}), the effective noise is close to Gaussian with the same covariance.  We then take the achievable scheme for the Gaussian noise, and construct an equivalent scheme with bounded output (using Lemma \ref{boundedoutputlemma}) and finite reading precision (using Lemma \ref{finiteprecisionlemma2}), and apply it to the new effective network with approximately Gaussian noise. 
Using the continuity property of finite reading precision schemes (Lemma \ref{continuouslemma}), and the property of bounded output, we conclude the proof by showing that the distortion achieved on this effective network with approximately Gaussian noise is close to what it would have been if the sources were actually Gaussian. 

\vspace{2mm}

\begin{proof}[Proof of Theorem \ref{thm2b}]
We start by noticing that, if the distortion tuple $(D_1,\ldots,D_k)$ is achievable in the case where $(Z_1[0],\ldots,Z_N[0])$ is jointly Gaussian with covariance matrix $\bv K$, then Lemmas \ref{boundedoutputlemma} and \ref{finiteprecisionlemma2} imply that we have a code $\C$ with block length $n$ with finite reading precision $\rho$,  
 which achieves distortion vector $(D_1 +\ep/2,\ldots,D_k+\ep/2)$
and for which
\aln{
\left\| g_{d_j} ( y_1,\ldots,y_n) \right\|_\infty \leq M,
}
for any $(y_1,\ldots,y_n) \in \R^n$, $j = 1,\ldots,k$ and a fixed $M > 0$.

In order to build our new coding scheme $\tilde \C$, we will first define operations that will be applied to transmit and received signals in the network.
For an integer $b$, these operations will effectively create $b$ new memoryless networks, which will replace the input-output relationship in (\ref{additiveio}) with
\al{ \label{effio}
\begin{bmatrix} \tilde Y_{1}^{(\ell)} \\ \tilde Y_{2}^{(\ell)} \\ \vdots \\ \tilde Y_{N}^{(\ell)} \end{bmatrix} & = 
{\bf H}
\begin{bmatrix} \tilde U_{1}^{(\ell)} \\ \tilde U_{2}^{(\ell)} \\ \vdots \\ \tilde U_{N}^{(\ell)} \end{bmatrix} + \begin{bmatrix} \tilde Z_{1}^{(\ell)} \\ \tilde Z_{2}^{(\ell)} \\ \vdots \\ \tilde Z_{N}^{(\ell)} \end{bmatrix},
}
for $\ell = 0,1,\ldots,b-1$, where $\tilde U_1^{(\ell)},\ldots,\tilde U_N^{(\ell)}$ are the effective inputs, $\tilde Y_1^{(\ell)},\ldots,\tilde Y_N^{(\ell)}$ the effective outputs, and $(\tilde Z_1^{(\ell)},\ldots,\tilde Z_N^{(\ell)})$ the effective additive-noise vector of the $\ell$th effective network.
In each of these memoryless networks we will then apply our original coding scheme $\C$.
The resulting coding scheme $\tilde C$ will thus have block length $nb$.

The operations applied to transmit signals and received signals will make use of the linear transformation $\bf Q$, defined in (\ref{qdef}).
Given a block of $nb$ effective inputs to a node $i \in V$ $\tilde U_i^{(\ell)} [0],\ldots,\tilde U_i^{(\ell)} [n-1]$, for $\ell = 0,1,\ldots,b-1$, the actual $nb$ transmit signals of node $i$ are built as
\al{ \label{effi}
\begin{bmatrix} U_{i}[t b] \\ U_{i}[t b+1] \\ \vdots \\ U_{i}[t b+b-1] \end{bmatrix} & = 
{\bf Q}^{-1}
\begin{bmatrix} \tilde U_{i}^{(0)}[t] \\ \tilde U_{i}^{(1)}[t] \\ \vdots \\ \tilde U_{i}^{(b-1)}[t] \end{bmatrix},
}
for $t= 0,1,\ldots,n-1$.
The effective network outputs are built from the actual received signals by setting
\al{ \label{effo}
\begin{bmatrix} \tilde Y_{i}^{(0)}[t] \\ \tilde Y_{i}^{(1)}[t] \\ \vdots \\ \tilde Y_{i}^{(b-1)}[t] \end{bmatrix} & = 
{\bf Q} \begin{bmatrix} Y_{i}[t b] \\ Y_{i}[t b+1] \\ \vdots \\ Y_{i}[t b+b-1] \end{bmatrix},
}
for $t= 0,1,\ldots,n-1$.
Using (\ref{effo}), (\ref{additiveio}) and then (\ref{effi}) we can write
\al{ 
& \begin{bmatrix} \tilde Y_{1}^{(0)}[t] & \tilde Y_{2}^{(0)}[t] & \cdots & \tilde Y_{N}^{(0)}[t]\\ 
\tilde Y_{1}^{(1)}[t] & \tilde Y_{2}^{(1)}[t] & \cdots & \tilde Y_{N}^{(1)}[t] \\ \vdots & \vdots & \ddots & \vdots \\ 
\tilde Y_{1}^{(b-1)}[t] & Y_{2}^{(b-1)}[t] & \cdots & Y_{N}^{(b-1)}[t] \end{bmatrix} = 
{\bf Q} 
\begin{bmatrix} Y_{1}[t b] & Y_{2}[t b] & \cdots & Y_{N}[t b] \\ 
Y_{1}[t b+1] & Y_{2}[t b+1] & \cdots & Y_{N}[t b+1] \\ \vdots & \vdots & \ddots & \vdots \\ 
Y_{1}[t b+b-1] & Y_{2}[t b+b-1] & \cdots & Y_{N}[t b+b-1] \end{bmatrix} \\
& \quad \quad = 
{\bf Q} 
\begin{bmatrix} U_1[tb] & \cdots & U_N[tb] \\ 
U_1[tb+1] & \cdots & U_N[tb+1]  \\ \vdots & \ddots & \vdots \\ 
U_1[tb+b-1] & \cdots & U_N[tb+b-1] \end{bmatrix} {\bf H}^T
+ {\bf Q} 
\begin{bmatrix} Z_1[tb] & \cdots & Z_N[tb] \\ 
Z_1[tb+1] & \cdots & Z_N[tb+1]  \\ \vdots & \ddots & \vdots \\ 
Z_1[tb+b-1] & \cdots & Z_N[tb+b-1] \end{bmatrix} \\
& \quad \quad =
\begin{bmatrix} \tilde U_{1}^{(0)}[t] & \cdots & \tilde U_{N}^{(0)}[t] \\ 
\tilde U_{1}^{(1)}[t] & \cdots & \tilde U_{N}^{(1)}[t]  \\ \vdots & \ddots & \vdots \\ 
\tilde U_{1}^{(b-1)}[t] & \cdots & \tilde U_{N}^{(b-1)}[t] \end{bmatrix} {\bf H}^T
+ {\bf Q} 
\begin{bmatrix} Z_1[tb] & \cdots & Z_N[tb] \\ 
Z_1[tb+1] & \cdots & Z_N[tb+1]  \\ \vdots & \ddots & \vdots \\ 
Z_1[tb+b-1] & \cdots & Z_N[tb+b-1] \end{bmatrix}.
}
This shows that we in fact have the effective input-output relationship shown in (\ref{effio}), if we define the effective noise vectors via
\al{ \label{effz}
\begin{bmatrix} \tilde Z_{1}^{(0)}[t] & \cdots & \tilde Z_{N}^{(0)}[t] \\ 
\tilde Z_{1}^{(1)}[t] & \cdots & \tilde Z_{N}^{(1)}[t]  \\ \vdots & \ddots & \vdots \\ 
\tilde Z_{1}^{(b-1)}[t] & \cdots & \tilde Z_{N}^{(b-1)}[t] \end{bmatrix} =
 {\bf Q} 
\begin{bmatrix} Z_1[tb] & \cdots & Z_N[tb] \\ 
Z_1[tb+1] & \cdots & Z_N[tb+1]  \\ \vdots & \ddots & \vdots \\ 
Z_1[tb+b-1] & \cdots & Z_N[tb+b-1] \end{bmatrix}.
}
In order to apply our original coding scheme $\C$ to the $\ell$th effective network, we set the effective transmit signals to be
\al{ \label{effcod1}
\tilde U_{s_m}^{(\ell)}[t] = f_{s_m,t}(X_m[\ell n:(\ell+1)n-1],\tilde Y_{s_m}^{(\ell)}[0:t-1])
}
for each source $s_m \in \Sscr$, and
\al{ \label{effcod2}
\tilde U_i^{(\ell)}[t] = f_{i,t}(\tilde Y_{i}^{(\ell)}[0:t-1])
}
for any other node $i \in \Rscr \cup \Dscr$.
The decoding functions are applied at each destination $d_m \in \Dscr$ as
\al{ \label{effcod3}
\hat X_m[\ell n:(\ell+1)n-1] = g_{d_m}(\tilde Y_{d_m}^{(\ell)}[0:n-1]),
}
and each destination can output the estimate $\hat X \in \R^{nb}$ by concatenating the length-$n$ output from each of the $b$ effective networks.
Next, we check that equations (\ref{effcod1}), (\ref{effcod2}) and (\ref{effcod3}) do not violate causality when they are mapped to the actual transmit and received signals according to (\ref{effi}) and (\ref{effo}).
From (\ref{effi}), we notice that, at time $tb$, for $t=0,\ldots,n-1$, in order to construct the next $b$ transmit signals $(U_i[tb],\ldots,U_i[tb+b-1])$, node $i$ needs the effective transmit signals $( \tilde U_i^{(0)}[t],\ldots,\tilde U_i^{(b-1)}[t])$.
These effective transmit signals, in turn, can be constructed given the effective received signals $\tilde Y_i^{(\ell)}[0],\ldots,\tilde Y_i^{(\ell)}[t-1]$, $\ell=0,\ldots,b-1$, from (\ref{effcod1}) and (\ref{effcod2}).
Finally, from (\ref{effo}), we notice that all these effective received signals can be constructed at the end of time slot $(t-1)b + b-1 = tb-1$, and will therefore be available at node $i$ at time $tb$. 
It can be similarly checked that the decoding operation defined in (\ref{effcod3}) does not violate causality.

We now analyze the effective noise vectors obtained.
The fact that the additive-noise vectors
\al{ \label{addnoisevec}
\left( \tilde Z_1^{(\ell)}[t], \tilde Z_2^{(\ell)}[t],\ldots,\tilde Z_N^{(\ell)}[t] \right)
}
for $t=0,\ldots,n-1$ are i.i.d.\,follows easily from the fact that they correspond to a row of the matrix on the left-hand side of (\ref{effz}), which is i.i.d.\,for $t=0,\ldots,n-1$, since the matrix on the right-hand side is i.i.d.~by the definition of the memoryless additive-noise network.
Moreover, by comparing (\ref{effz}) with (\ref{tildex}), we see that Lemma \ref{convlemma} implies that
\al{
\left( \tilde Z_1^{(\ell_b)}[t], \tilde Z_2^{(\ell_b)}[t],\ldots,\tilde Z_N^{(\ell_b)}[t] \right) \stackrel{d}{\to} \Nscr({\bf 0}, {\bf K}),
}
as $b \to \infty$, for each $t \in \{0,\ldots,n-1\}$, and any sequence $\ell_b$, $b=1,2,...$, such that $\ell_b \in \{0,\ldots,b-1\}$.

From Lemma \ref{finiteprecisionlemma2}, we were able to assume that the original coding scheme $\C$ has finite reading precision, which implies that the encoding, relaying and decoding functions $f_{s_m,t}$, $f_{r_p,t}$, $f_{d_m,t}$ and $f_{d_m,t}$ are continuous almost everywhere.
It is then not difficult to see that, for each destination $d_m$, we can write 
\aln{
\left( \tilde Y_{d_m}^{(\ell)}[0:n-1] \right) = F_m\left( \vec{X}_1^{(\ell)},\ldots,\vec{X}_k^{(\ell)}, \vec{\tilde Z}_1^{(\ell)},\ldots,\vec{\tilde Z}_N^{(\ell)} \right),
}
where $F_m$ is an almost-everywhere-continuous function of $\vec{\tilde Z}_1^{(\ell)},\ldots,\vec{\tilde Z}_N^{(\ell)}$,
 $\vec{X}_i^{(\ell)} = X_i[\ell n :\ell (n+1)-1]$ and $\vec{\tilde Z}_i^{(\ell)} = \tilde Z_i^{(\ell)}[0:n-1]$, for $\ell = 0,\ldots,b-1$.
%
%
%
%
%
Therefore, the mapping
\aln{
\left\{\vec{\tilde Z}_1^{(\ell)},\ldots,\vec{\tilde Z}_N^{(\ell)} \right\}   \mapsto  \left\| {\bf X}_m^{(\ell)} - g_{d_m}\left( F_m\left( \vec{X}_1^{(\ell)},\ldots,\vec{X}_k^{(\ell)}, \vec{\tilde Z}_1^{(\ell)},\ldots,\vec{\tilde Z}_N^{(\ell)} \right) \right)  \right\|^2,
}
for $m=1,\ldots,k$, is continuous almost everywhere as well.
We conclude that 
\aln{ 
&  \left\| {\bf X}_m^{(\ell_b)} - g_{d_m}\left( F_m\left( \vec{X}_1^{(\ell_b)},\ldots,\vec{X}_k^{(\ell_b)}, \vec{\tilde Z}_1^{(\ell_b)},\ldots,\vec{\tilde Z}_N^{(\ell_b)} \right) \right)  \right\|^2 \\
&   = \left\| {\bf X}_m^{(0)} - g_{d_m}\left( F_m\left( \vec{X}_1^{(0)},\ldots,\vec{X}_k^{(0)}, \vec{\tilde Z}_1^{(\ell_b)},\ldots,\vec{\tilde Z}_N^{(\ell_b)} \right) \right)  \right\|^2 \\
&  \stackrel{d}{\to} 
\left\| {\bf X}_m^{(0)} - g_{d_m}\left( F_m\left( {\bf X}_1^{(0)},\ldots,{\bf X}_k^{(0)}, \vec{\tilde Z}_1,\ldots,\vec{\tilde Z}_N \right) \right)  \right\|^2,
}
as $b \to \infty$, 
where $\vec{\tilde Z}_i = \tilde Z_i[0:n-1]$ for $i=1,\ldots,N$, and $\left\{ \left( \tilde Z_1[t],\ldots,\tilde Z_N[t] \right) \right\}_{t=0}^{n-1}$ is an i.i.d.\,sequence of jointly Gaussian vectors $\Nscr({\bf 0},{\bf K})$.
Moreover, we have that 
\al{ 
& \E \left\| {\bf X}_m^{(0)} - g_{d_m}\left( F_m\left( \vec{X}_1^{(0)},\ldots,\vec{X}_k^{(0)}, \vec{\tilde Z}_1^{(\ell_b)},\ldots,\vec{\tilde Z}_N^{(\ell_b)} \right) \right)  \right\|^2 \nonumber \\
& \quad \quad \leq 2 \E \left( \left\| {\bf X}_m^{(0)} \right\|^2 +  \left\| g_{d_m}\left( F_m\left( \vec{X}_1^{(0)},\ldots,\vec{X}_k^{(0)}, \vec{\tilde Z}_1,\ldots,\vec{\tilde Z}_N \right) \right)  \right\|^2 \right) \nonumber \\
& \quad \quad \leq 2 \E \left( \left\| {\bf X}_m^{(0)} \right\|^2 + n M^2 \right) \nonumber \\
& \quad \quad = 2 n {\bf K}_{m,m} + 2 n M^2 < \infty.   \label{bound21}
}
Thus, from a variant of the Dominated Convergence Theorem (see Appendix \ref{appendixH}), we conclude that, as $b\to\infty$,
\al{ 
&  \E \left\| {\bf X}_m^{(\ell_b)} - g_{d_m}\left( F_m\left( \vec{X}_1^{(\ell_b)},\ldots,\vec{X}_k^{(\ell_b)}, \vec{\tilde Z}_1^{(\ell_b)},\ldots,\vec{\tilde Z}_N^{(\ell_b)} \right) \right)  \right\|^2  \nonumber \\
& \quad \quad \quad  \stackrel{d}{\to} 
\E \left\| {\bf X}_m^{(0)} - g_{d_m}\left( F_m\left( \vec{X}_1^{(0)},\ldots,\vec{X}_k^{(0)}, \vec{\tilde Z}_1,\ldots,\vec{\tilde Z}_N \right) \right)  \right\|^2. \label{expconv}
}
Therefore, we can choose $b$ sufficiently large so that 
\al{ 
& \frac1n \E \left\| {\bf X}_m^{(\ell_b)} - g_{d_m}\left( F_m\left( \vec{X}_1^{(\ell_b)},\ldots,\vec{X}_k^{(\ell_b)}, \vec{\tilde Z}_1^{(\ell_b)},\ldots,\vec{\tilde Z}_N^{(\ell_b)} \right) \right)  \right\|^2 \nonumber \\ 
& \quad \quad \quad \quad \leq \frac1n \E \left\| {\bf X}_m^{(0)} - g_{d_m}\left( F_m\left( \vec{X}_1^{(0)},\ldots,\vec{X}_k^{(0)}, \vec{\tilde Z}_1,\ldots,\vec{\tilde Z}_N \right) \right)  \right\|^2 + \ep/2 \nonumber \\
& \quad \quad \quad \quad \leq D_m + \ep. \label{bound22}
}
The expected distortion of code $\tilde \C$ (with blocklength $nb$) thus satisfies
\aln{
& \frac1{nb} \sum_{\ell = 0}^{b-1} \E \left\| {\bf X}_m^{(\ell)} - g_{d_m}\left( F_m\left( \vec{X}_1^{(\ell)},\ldots,\vec{X}_k^{(\ell)}, \vec{\tilde Z}_1^{(\ell)},\ldots,\vec{\tilde Z}_N^{(\ell)} \right) \right)  \right\|^2 \\ 
& \quad \quad \quad \quad \leq \frac1n \max_{0 \leq \ell \leq b-1}\E \left\| {\bf X}_m^{(\ell)} - g_{d_m}\left( F_m\left( \vec{X}_1^{(\ell)},\ldots,\vec{X}_k^{(\ell)}, \vec{\tilde Z}_1^{(\ell)},\ldots,\vec{\tilde Z}_N^{(\ell)} \right) \right)  \right\|^2 \\
& \quad \quad \quad \quad \leq D_m + \ep,
}
for $m=1,\ldots,k$, since (\ref{expconv}) holds in particular for the sequence 
\aln{
\ell_b = \arg \max_{0 \leq \ell \leq b-1}\E \left\| {\bf X}_m^{(\ell)} - g_{d_m}\left( F_m\left( \vec{X}_1^{(\ell)},\ldots,\vec{X}_k^{(\ell)}, \vec{\tilde Z}_1^{(\ell)},\ldots,\vec{\tilde Z}_N^{(\ell)} \right) \right)  \right\|^2,
}
for $b=1,2,...$.
This concludes the proof of Theorem \ref{theorem2}. 
\end{proof}
\begin{note}
Note that the setup in this section (Theorem \ref{theorem2}, or equivalently Theorem \ref{thm2b}) assumes quadratic distortion as the criterion of distortion between the source and its reconstruction. Indeed, the arguments in the proofs of some of the supporting lemmas above do depend on the nature of the distortion metric. 
However, the overall idea of transforming the source or the channel into approximately Gaussian, and then constructing encoding-decoding schemes with finite encoding or reading precision is independent of the nature of the distortion metric. It can be shown that, under mild conditions, our results carry over to other distortion metrics. 
This in general is not true for the setup in Theorem \ref{theorem1} (or equivalently Theorem \ref{thm1b}), where the proof  hinges, among other things, on the fact that the distortion in the transform and the original domain is the same.
\end{note}


\section{Conclusion}
\label{sec::conclusion}

We considered the problem of distributed compression of correlated sources over a network, for which we established two complementary worst-case results.
The first one is that, under a source covariance constraint, the worst-case source is Gaussian.
The second is that, for  additive-noise networks where the noises satisfy a covariance constraint, the worst-case noise is also Gaussian.
These results provide a theoretical justification for the common adoption of Gaussian models for both the source and the additive noise in distributed  compression problems.

Our approach to establish these results is constructive, in that we describe a systematic way of converting coding schemes designed under Gaussian assumptions into coding schemes that can handle non-Gaussian assumptions. 
The idea behind the construction of such schemes is simple both conceptually and algorithmically, as the DFT transform can be implemented via the tractable FFT, and the remaining part is to employ a good scheme for the multi-terminal Gaussian source or channel, for which there is a well-developed and growing body of literature.

Another interesting aspect of our framework for converting coding schemes designed under Gaussian assumptions into coding schemes for non-Gaussian models is that it only requires the mean and the covariance matrix of the sources or the additive noises. 
Thus, this code conversion scheme can be seen as a way to design coding schemes for distributed compression problems where only the mean and the covariance matrix of the sources or the additive noises are known.
Therefore, this work may provide tools for future research in establishing inner bounds in the distortion region of distributed compression problems with unknown source or noise distributions.

Another possible research direction stemming from this work concerns finding outer bounds to the distortion region of Gaussian problems.
Notice that, even for the Gaussian $k$-source encoding problem, the rate-distortion region is unknown for $k>2$, and finding nontrivial outer bounds is in general difficult. 
In this work, we showed that, given the appropriate covariance constraints, Gaussian sources and Gaussian additive noises are worst-case assumptions.
This means that the distortion region under non-Gaussian assumptions contains the Gaussian distortion region.
Thus, by choosing special source and noise distributions (e.g., discrete distributions), it may be possible to obtain distributed compression problems where interesting outer bounds can be derived.
Our results would then imply that any outer bound derived in this manner is also an outer bound on the distortion region of the corresponding Gaussian problem.
\bibliographystyle{unsrt}



\begin{thebibliography}{1}

\bibitem{LapidothNearest}
A.~Lapidoth.
\newblock Nearest neighbor decoding for additive non-gaussian noise channels.
\newblock {\em IEEE Transactions on Information Theory}, 42(5):1520--1529,
  September 1996.

\bibitem{DiggaviWorstCase}
S.~N. Diggavi and T.~M. Cover.
\newblock The worst additive noise under a covariance constraint.
\newblock {\em {IEEE Transactions on Info. Theory}}, 47(7):3072--3081, November
  2001.

\bibitem{ShamaiWorstCase}
S.~Shamai and S.~Verdu.
\newblock Worst-case power-constrained noise for binary-input channels.
\newblock {\em IEEE Transactions on Information Theory}, 38(5):1494--1511,
  September 1992.

\bibitem{Aaron2Terminal}
A.~B. Wagner, S.~Tavildar, and P.~Viswanath.
\newblock Rate region of the quadratic {{Gaussian}} two-encoder source-coding
  problem.
\newblock {\em IEEE Transactions on Information Theory}, 54(5):1938--1961, May
  2008.

\bibitem{wcnoisefull}
I.~Shomorony and A.~S. Avestimehr.
\newblock Worst-case additive noise in wireless networks.
\newblock {\em to appear in IEEE Transactions on Information Theory}, 2012.

\bibitem{wcsourceITW}
I.~Shomorony, A.~S. Avestimehr, H.~Asnani, and T.~Weissman.
\newblock Worst-case source for distributed compression with quadratic
  distortion.
\newblock {\em In Proc. of Information Theory Workshop (ITW)}, 2012.

\bibitem{billingsley}
P.~Billingsley.
\newblock {\em Probability and Measure}.
\newblock Wiley Series in Probability and Mathematical Statistics. John Wiley
  \& Sons, 3rd edition, 1995.

\end{thebibliography}

\appendices

\section{Proof of Lemma \ref{convlemma}}
\label{appendixG}
Clearly, it suffices to show that $\left(\til X_1^{(\ell_b)}[0],...,\til X_k^{(\ell_b)}[0]\right)$ converges in distribution to a jointly Gaussian random vector with covariance matrix $\vec K$, as $b \to \infty$.
In order to use the Cram\'er-Wold Theorem \cite{billingsley}, we fix an arbitrary vector $(t_1,...,t_k) \in \R^k$ and we notice that
\al{ 
\sum_{m=1}^k t_m \til X_m^{(\ell_b)}[0] 
& = \sum_{m=1}^k t_m \sum_{j=0}^{b-1} X_m[j] ~ Q(\ell_b,j) \nonumber \\
& = \sum_{j=0}^{b-1} \left( \sum_{m=1}^k t_m X_m[j] \right) Q(\ell_b,j). \label{cramerexp}
}
To characterize the convergence in distribution of (\ref{cramerexp}), we will need the following result.
\begin{theorem}[Lindeberg's Central Limit Theorem \cite{billingsley}] \label{lindthm}
Suppose that for each $b = 1,2,...$, the random variables
$
Y_{b,1}, Y_{b,2},..., Y_{b,b}
$
are independent.
In addition, suppose that, for all $b$ and $i \leq b$, $E[Y_{b,i}] = 0$, and let
\al{
s_b^2 = \sum_{i=1}^b E\left[ Y_{b,i}^2 \right].   \label{sbdef}
} 
Then, if for all $\vep > 0$, Lindeberg's condition
\al{
\frac{1}{s_b^2}\sum_{i=1}^b E\left(Y_{b,i}^2 \, \one\left\{|Y_{b,i}|  \geq \vep s_b\right\}\right) \goesto 0 \text{ as $b \goesto \infty$}
\label{lind}
}
holds, we have that
\aln{
\frac{\sum_{i=1}^b Y_{b,i}}{s_b} \stackrel{d}{\goesto} \Nscr(0,1).
}
\end{theorem}
To apply Lindeberg's CLT, we will let, for $j=0,...,b-1$,
\aln{
Y_{b,j+1} = \sqrt b \left( \sum_{m=1}^k t_m X_m[j] \right) Q(\ell_b,j).
}
Then, if we let $\vec K_{u,v}$ be the entry in the $u$th row and $v$th column of $\vec K$, we have
\aln{
s_b^2 & = \sum_{j=1}^b E\left[ Y_{b,j}^2 \right] = b \sum_{j=1}^b  Q^2(\ell_b,j-1)  E \left( \sum_{m=1}^k t_m X_m[j-1] \right)^2 \\
& = b \sum_{1 \leq u,v \leq k} t_u t_v \vec K_{u,v}  \sum_{j=1}^b  Q^2(\ell_b,j-1)  \\
& = b \sum_{1 \leq u,v \leq k} t_u t_v \vec K_{u,v},
} 
regardless of the value of $\ell_b$.
In order to verify Lindeberg's condition, we define $\sigma^2 =\sum_{1 \leq u,v \leq k} t_u t_v \vec K_{u,v}$ and we let $U_{b,j} = Y_{b,j}^2 \, \one \left\{|Y_{b,j}|  \geq \vep s_b \right\} = Y_{b,j}^2 \, \one \left\{|Y_{b,j}|  \geq \vep \sigma \sqrt b \right\}$.
Consider any sequence $j_b$, for $b=1,2,...$, such that $j_b \in \{1,...,b\}$, and any $\delta > 0$.
Then we have that
\aln{
\Pr\left(U_{b,j_b} < \delta\right) & \geq \Pr\left(|Y_{b,j_b}| < \vep \sigma \sqrt{b} \right) \geq \Pr\left(\left|\sum_{m=1}^k t_m X_m[j_b-1]\right|\sqrt2 < \vep \sigma \sqrt{b} \right) \\
& = \Pr\left(\left|\sum_{m=1}^k t_m X_m[0]\right| < \vep \sigma \sqrt{b/2} \right) \goesto 1, 
}
as $b \goesto \infty$, which means that $U_{b,j_b} \stackrel{p}{\goesto} 0$ (i.e., $U_{b,j_b}$ converges in probability to $0$) as $b \goesto \infty$.
Moreover, we have that,
\aln{
|U_{b,j_b}|\le Y_{b,j_b}^2\le 2 \left(\sum_{m=1}^k t_m X_m[j_b-1]\right)^2
}
for $b=1,2,...$,
and
\aln{
\E \left[ 2 \left(\sum_{m=1}^k t_m X_m[0]\right)^2 \right] = 2 \sigma^2 < \infty.
}
Thus by the Dominated Convergence Theorem (see Appendix \ref{appendixH}), we have that $E[U_{b,j_b}]\goesto 0$ as $b \to \infty$.
We conclude that
\aln{
\frac{1}{s_b^2}\sum_{i=1}^b E\left(Y_{b,j}^2 \, \one\left\{|Y_i|  \geq \vep s_b\right\}\right) & = \frac{1}{\sigma^2 b} \sum_{j=1}^b E\left[ U_{b,j} \right] \\
& \leq \frac{1}{\sigma^2} \max_{1 \leq j \leq b} E\left[ U_{b,j} \right]
 \to 0, 
}
as $b \to \infty$, and Lindeberg's condition (\ref{lind}) is satisfied for any $\vep > 0$.
Hence, from Theorem \ref{lindthm}, we have that
\aln{
\frac{\sum_{i=1}^{b} Y_{b,j}}{\sigma \sqrt{b}} \stackrel{d}{\goesto} \Nscr(0,1), 
} 
which implies, from (\ref{cramerexp}), that
\aln{
\sum_{m=1}^k t_m \til X_m^{(\ell_b)}[0] & =  \sum_{j=0}^{b-1} \left( \sum_{m=1}^k t_m X_m[j] \right) Q(\ell_b,j) \\
& = \frac{\sum_{j=1}^{b} Y_{b,j}}{\sqrt{b}} \stackrel{d}{\goesto} \mathcal{N}(0,\sigma^2).
}
Finally, since for a jointly Gaussian vector $(X^G_1,...,X^G_k)$ with mean zero and covariance matrix $\vec K$, we have $\sum_{m=1}^k t_m X^G_m \sim \Nscr(0,\sigma^2)$, we conclude, from the Cram\'er-Wold Theorem that $\left(\til X_1^{(\ell_b)}[0],...,\til X_k^{(\ell_b)}[0]\right)$ converges in distribution to a jointly Gaussian random vector with zero mean and covariance matrix $\vec K$, as $b \to \infty$.

\section{Proof of Lemma \ref{boundedoutputlemma}}
\label{appendixA}
From a coding scheme $\C$ with blocklength $n$ achieving distortion vector $(D_1,...,D_k)$, we will create a sequence of coding schemes $\C^{(m)}$, $m=1,2,...$, obtained by clipping the output of the decoding functions $g_{d_j}$, $j=1,...,k$.
More precisely, coding scheme $\C^{(m)}$ has the same encoding and relaying functions as $\C$, and decoding functions $g_{d_j}^{(m)}$ whose $i$th component is defined as
\aln{
g_{d_j}^{(m)} (y_1,...,y_n)[i] = \left\{
\begin{array}{ll}
m, & \, \text{  if $g_{d_j}(y_1,...,y_n)[i] > m$} \\
-m, & \, \text{  if $g_{d_j}(y_1,...,y_n)[i] < -m$} \\
g_{d_j}(y_1,...,y_n)[i],  & \, \text{  otherwise} \\
\end{array} \right.
}
for $j=1,...,k$, and $i=0,...,n-1$.
Now, consider a fixed $j \in \{1,...,k\}$, and define, for $i = 0,...,n-1$, the event $B_i$ as
\aln{
B_i = \left\{ X_j[i] > m, g_{d_j}\left( Y_{d_j}^n \right)[i] > m \right\} \cup \left\{ X_j[i] < -m, g_{d_j}\left( Y_{d_j}^n \right)[i] < -m \right\}.
}
It is easy to verify that the complementary event is given by
\aln{
B_i^c = \left\{ \left| X_j[i] \right| \leq m \right\} \cup \left\{ \left| g_{d_j}\left( Y_{d_j}^n \right)[i] \right| \leq m \right\} \cup \left\{ X_j[i] > m, g_{d_j}\left( Y_{d_j}^n \right)[i] < -m \right\} \cup \left\{ X_j[i] < -m, g_{d_j}\left( Y_{d_j}^n \right)[i] > m \right\}.
}
For each of the four sub-events in $B_i^c$, it is clear that 
\aln{
\left| X_j[i] - g_{d_j}\left( Y_{d_j}^n \right)[i] \right| \geq \left| X_j[i] - g^{(m)}_{d_j}\left( Y_{d_j}^n \right)[i] \right|.
}
Thus, we can upper bound the expected distortion of the output of decoder $j$ of $\C^{(m)}$ as
\aln{
\E\left[ \left\| X_j^n - g_{d_j}^{(m)}\left( Y_{d_j}^n \right) \right\|^2 \right]
& = \sum_{i=0}^ {n-1} \E\left[ \left( X_j[i] - g_{d_j}^{(m)}\left( Y_{d_j}^n \right) [i] \right)^2 \right] \\
& = \sum_{i=0}^ {n-1} \left\{ \E\left[ \left( X_j[i] - g_{d_j}^{(m)}\left( Y_{d_j}^n \right) [i] \right)^2 \one_{B_i^c} \right] + \E\left[ \left( X_j[i] - g_{d_j}^{(m)}\left( Y_{d_j}^n \right) [i] \right)^2 \one_{B_i} \right] \right\} \\
& \leq \sum_{i=0}^ {n-1} \left\{ \E\left[ \left( X_j[i] - g_{d_j}\left( Y_{d_j}^n \right) [i] \right)^2  \right] + \E\left[ \left( X_j[i] - g_{d_j}^{(m)}\left( Y_{d_j}^n \right) [i] \right)^2 \one_{B_i} \right] \right\} \\
& = \E\left[ \left\| X_j^n - g_{d_j}\left( Y_{d_j}^n \right) \right\|^2 \right] + \sum_{i=0}^ {n-1} \E\left[ \left( X_j[i] - g_{d_j}^{(m)}\left( Y_{d_j}^n \right) [i] \right)^2 \one_{B_i} \right] \\
& \leq \E\left[ \left\| X_j^n - g_{d_j}\left( Y_{d_j}^n \right) \right\|^2 \right] + \sum_{i=0}^ {n-1} \E\left[ \left( X_j[i] \right)^2 \one_{B_i} \right] \\
& = n D_j + n \E\left[ \left( X_j[0] \right)^2 \one_{B_0} \right].
}
Since $\left| X_j[0]^2 \one_{B_0} \right| \leq X_j[0]^2$, $E\left[ X_j[0]^2 \right] < \infty$, and
$X_j[0]^2 \one_{B_0} \stackrel{p}{\to} 0$ as $m \to \infty$, by the Dominated Convergence Theorem (see  Appendix \ref{appendixH}), 
\aln{
\lim_{m \to \infty} \E\left[ \left( X_j[0] \right)^2 \one_{B_0} \right] = 0.
}
Therefore, for any $\ep > 0$, we can pick $m = M$ large enough so that 
\aln{
\tfrac1n \E\left[ \left\| X_j^n - g_{d_j}^{(M)}\left( Y_{d_j}^n \right) \right\|^2 \right] \leq D_j + \ep
\quad \quad \text{and} \quad \quad 
\left\| g_{d_j}^{(M)} ( y_1,...,y_n) \right\|_\infty \leq M,
}
for all $j = 1,...,K$, and we may let $\tilde \C = \C^{(M)}$.

\section{Proof of Lemma \ref{densitylemma}}
\label{appendixB}
For the sake of simplicity, we will consider the case $k=2$ and $i = 2$ (i.e, $Y_1$ is quantized to $\tilde{Y}_1$).
The proof for $k>2$ follows via a straightforward generalization. The proof follows similar lines of thought as Lemma 3 in \cite{wcnoisefull}, we state here the required steps for completeness. The density $f_{\tilde{Y}_1,Y_2}(y_1,y_2)$ can be written for almost all tuples $(y_1,y_2)$ as,
\bea
f_{\tilde{Y}_1,Y_2}(y_1,y_2)&=&2^{\rho}\E[\1_{\{y_1-\lfloor Y_1\rfloor_{\rho}\in(-2^{-\rho-1},2^{-\rho-1})\}}|Y_2= y_2] \; f_{Y_2}(y_2) \nonumber \\
&=&2^{\rho}\Pr[y_1-\lfloor Y_1\rfloor_{\rho}\in(-2^{-\rho-1},2^{-\rho-1})|Y_2= y_2] \;  f_{Y_2}(y_2) \nonumber \\
&=&2^{\rho}\Pr[\lfloor Y_1\rfloor_{\rho}\in(y_1-2^{-\rho-1},y_1+2^{-\rho-1})|Y_2= y_2] \; f_{Y_2}(y_2) \nonumber \\
&=&2^{\rho}\Pr[\lfloor 2^{\rho}Y_1\rfloor\in(y_12^{\rho}-\frac{1}{2},y_12^{\rho}+\frac{1}{2})|Y_2= y_2] \; f_{Y_2}(y_2) \nonumber \\
&=&2^{\rho}\Pr[2^{\rho}Y_1\in(\lceil y_12^{\rho}-\frac{1}{2}\rceil,\lceil y_12^{\rho}+\frac{1}{2}\rceil)|Y_2= y_2] \; f_{Y_2}(y_2) \nonumber \\
&=&2^{\rho}\Pr[Y_1\in(2^{-\rho}\lceil y_12^{\rho}-\frac{1}{2}\rceil,2^{-\rho}\lceil y_12^{\rho}+\frac{1}{2}\rceil)|Y_2= y_2] \; f_{Y_2}(y_2) \nonumber \\
&=&2^{\rho} \int_{a_{\rho}}^{b_{\rho}}f_{Y_1,Y_2}(x_1,y_2)dx_1,
\eea
where $a_{\rho}=2^{-\rho}\lceil y_12^{\rho}-\frac{1}{2}\rceil$ and $b_{\rho}=2^{-\rho}\lceil y_12^{\rho}+\frac{1}{2}\rceil$, such that $b_{\rho}=a_{\rho}+2^{-\rho}$ which implies, $a_{\rho}\rightarrow y_1$. 
What is left to prove is that 
\aln{
\lim_{\rho\rightarrow\infty}2^{\rho} \int_{a_{\rho}}^{b_{\rho}}f_{Y_1,Y_2}(x_1,y_2)dx_1 =f_{Y_1,Y_2}(y_1,y_2)
}
for almost all tuples $(y_1,y_2)$. But this follows using the proof of Lemma 3 in \cite{wcnoisefull},  replacing the integrand function appropriately. 

\section{Proof of Lemma \ref{frlemma}}
\label{appendixC}

We prove the lemma by induction on the size $t$ of the random vector $Y$.
If $Y$ is a scalar, i.e., $t = 1$, let $g_u(y) = F_{Y|U}(y|u)$, where $F_{Y|U}$ is the conditional distribution function of $Y$ given $U$.
Then we let $Q$ be a uniform random variable on $[0,1]$ (independent of $U$), and we let $h(u,q) = g_u^{-1}(q)$ (where $^{-1}$ represents the generalized inverse).
It is then clear that $h(u,Q)$ is distributed as $Y$ conditioned on $U = u$ for any $u$, which implies that $(h(U,Q),U)$ is distributed as $(Y,U)$.

Now suppose the lemma is true when the size of $Y$ is $t$.
Consider a random vector $Y' = (Z^{t},\tilde Y)$, where $Z^{t}$ has size $t$ and $\tilde Y$ is a scalar.
Then there exists a random vector $Q'$ and a function $h'$ such that $(h'(U,Q'),U)$ is distributed as $(Z^t,U)$.
Now let $g_{u,z^t}(y) = F_{\tilde Y|U,Z^t}(y|u,z^t)$ be the conditional distribution function of $\tilde Y$ given $U$ and $Z^t$.
Then we let $Q = (Q',Q'')$, where $Q''$ is a uniform random variable on $[0,1]$ (independent of $U$ and $Q'$), and we let $h(u,(q',q'')) = g_{u,h'(u,q')}^{-1}(q'')$.
Then $(h(U,Q),U)$ is distributed as $(\tilde Y,U)$, and $(h'(U,Q'),h(U,Q),U)$ is distributed as $(Y',U) = (Z^t,\tilde Y,U)$.


\section{Proof of Lemma \ref{finiteprecisionlemma}}
\label{appendixD}
Achievability of the distortion tuple $(D_1,\cdots,D_k)$ implies the existence of a
coding scheme $\mathcal{C}$ with block length $n$, such that,
\bea
\frac{1}{n}\E\Big{[}\parallel
\mathbf{X}_m-\hat{\mathbf{X}}_m\parallel^2\Big{]}\le D_m,\ \forall\ m=[1:k].
\eea
Using Lemma \ref{boundedoutputlemma}, without loss of generality we will suppose that,  
\aln{
\left\| g_{d_j} ( y_1,...,y_n) \right\|_\infty \leq M,
}
for each destination $d_j \in \D$, for a fixed $M>0$. Note that, using Lemma \ref{frlemma}, the
memoryless channel $f_{Y_1,\cdots,Y_N|U_1,\cdots,U_N}$ can be equivalently
represented as a deterministic channel $Y_i=h_i(U_1,\cdots,U_N,{\bf Z}),\ \forall
i=[1:N]$ where ${\bf Z}$ is a random vector, independent of the
channel inputs, $(U_1,\cdots,U_N)$. Thus for a fixed block length $n$, given
the description of our encoding procedure, we can write, for some functions
$F_i$ depending on
$h_i$,
$\mathbf{Y}_i=F_i(\mathbf{X}_{1},\mathbf{X}_{2},\cdots,\mathbf{X}_{k},
\mathbf {Z}),\ \forall\ i\in [1:N]$, as the evolution of the system depends
only on the sources 
and the random vector ${\bf Z}$. Thus, noting that the reconstruction for the
$m${th} source is $\mathbf{\hat{X}}_m=g_{d_m}(\mathbf{Y}_m)$, the above equation on
distortion constraints can be equivalently written as,
\bea
\frac{1}{n}\E\Big{[}\parallel
\mathbf{X}_m-g_{d_m}(F_m(\mathbf{X}_1,\cdots,\mathbf{X}_k,\mathbf{Z}
))\parallel^2\Big { ] } \le D_m,\ \forall\ m=[1:k], 
\eea
To prove this lemma we have to show that, given an $\epsilon>0$, we can construct a scheme $\mathcal{C}_{\rho}$ for some $\rho=[\rho_1,\cdots,\rho_k]\in\N^k$,  where the encoding function at each source $s_m \in \mathcal{S}$ satisfies
\aln{
\tilde{f}_{s_m,t} ( x_m^n, y^{t-1}) = \tilde{f}_{s_m,t} ( \left\lfloor x_m^n \right\rfloor_{\rho_m}, y^{t-1}),\ \forall\ m\in[1:k]
}
for any $x_m^n \in \R^n$, any $y^{t-1} \in \R^{t-1}$, and any time $t$, such that, 
 
\bea
\frac{1}{n}\E\Big{[}\parallel
\mathbf{X}_m-g_{d_m}(F_m(\lfloor\mathbf{X}_1\rfloor_{\rho_1},\cdots,\lfloor\mathbf{X
}_k\rfloor_{\rho_k} , \mathbf { Z}
))\parallel^2\Big { ] } \le D_m+\epsilon,\ \forall\ m=[1:k], 
\eea
To prove this, we consider the following randomized encoding scheme
$\mathcal{C}_{\rho}$. 
Note the disclaimer that, in our definition of schemes,
the encoding, relaying and decoding operations were defined to be
deterministic, but for the time being we will allow for randomization and later
show that it can be dispensed with. The scheme $\mathcal{C}_{\rho}$, operated
in blocks of length $n$, uses the same relaying encoding and destination
encoding and decoding functions, the only change being in the source encoding. At
the source node $s_m$ the source is encoded as,
$U_{s_m,t}=f_{s_m,t}(\mathbf{\tilde{X}}_m,Y^{t-1})$, $\forall\ 
t\in[1:N]$, where $\mathbf{\tilde{X}}_m=\{\mathbf{\tilde{X}}_m[t]\}_{t=0}^{n-1}$,
such that $\mathbf{\tilde{X}}_m[t]=\lfloor
\mathbf{X}_m[t]\rfloor_{\rho_m}+V_{\rho_m}$, where $V_{\rho_m}$ is a random
variable independent of the sources in the network, uniformly distributed in
$(-2^{-\rho_m-1},2^{-\rho_m-1})$. 
Consider
\bea
\label{eq1precision}
&&\frac{1}{n}
\E\Big { [ } \parallel
\mathbf{X}_m-g_{d_m}(F_m(\mathbf{\tilde{X}}_1,\cdots,
\mathbf { \tilde{X}
}_k, \mathbf { Z }
))\parallel^2\Big { ] } \nonumber \\
&\le&\underbrace{\frac{1}{n}
\E\Big { [ } \parallel
\mathbf{X}_m-\mathbf{\tilde{X}}_m\parallel^2\Big { ]
}}_{(I)}+\underbrace{\frac{1}{n}
\E\Big { [ } \parallel
\mathbf{\tilde{X}}_m-g_{d_m}(F_m(\mathbf{\tilde{X}}_1,\cdots,
\mathbf { \tilde{X}
}_k , \mathbf { Z }
))\parallel^2\Big { ]
}}_{(II)}\nonumber\\
&&+\underbrace{\frac{1}{n}\E\Big{[}\parallel\mathbf{X}_m-\mathbf{\tilde{X}}
_m\parallel \parallel \mathbf { \tilde { X } }
_m-g_{d_m}(F_m(\mathbf{\tilde{X}}_1,\cdots,
\mathbf { \tilde{X}
}_k , \mathbf { Z }
))\parallel \Big{]}}_{(III)}.
\eea
Note that
\bea
|\mathbf{X}_m[t]-\mathbf{\tilde{X}}_m[t]|&=&|-V_{\rho_m}+\mathbf{X}_m[t]
-\lfloor\mathbf{X}_m[t]\rfloor_{\rho_m}| \nonumber \\
&=&|-V_{\rho_m}+2^{-\rho_m}(2^{\rho_m}\mathbf{X}_m[t]
-\lfloor2^{\rho_m}\mathbf{X}_m[t]\rfloor ) | \nonumber \\
&\le&|V_{\rho_m}|+2^{-\rho_m}|2^{\rho_m}\mathbf{X}_m[t]
-\lfloor2^{\rho_m}\mathbf{X}_m[t]\rfloor| \nonumber \\
&\le& 2^{-\rho_m-1}+2^{-\rho_m} \nonumber \\
&\le& 2^{-\rho_m+1},
\eea
which implies $\parallel \mathbf{X}_m-\mathbf{\tilde{X}}_m\parallel \le
\sqrt{n}2^{1-\rho_m}$. This further implies that the term (I) of (\ref{eq1precision}) is bounded as
\bea
\frac{1}{n}
\E\Big { [ } \parallel
\mathbf{X}_m-\mathbf{\tilde{X}}_m\parallel^2\Big { ]
}\le 2^{2-2\rho_m},
\eea
implying that, in the limit, term (I) vanishes. 
Define the (measurable) functions
$\mathbf{H}_m(\cdots) :
\underbrace{\mathbf{R}^n\times\cdots\times\mathbf{R}^n}_{\mbox{$k+1$
times}}\rightarrow \mathbf{R}, \forall\ m\in[1:k]$ as \bea
\mathbf{H}_m(\mathbf{{y}}_1,\cdots,\mathbf{{y}}_k,\mathbf{z})=\parallel
\mathbf{y}_m-g_{d_m}(F_m(\mathbf{{y}}_1,\cdots,
\mathbf { {y}
}_k, \mathbf { z }
))\parallel.
\eea

Since $\mathbf{Z}$ is independent of the sources, using Lemma \ref{densitylemma}, we have the following convergence of the joint densities,
\bea
\lim_{\rho_m\rightarrow\infty}
f(\mathbf {
X } _1 , \cdots
,\mathbf{\tilde{X}}_m,\cdots,\mathbf{X}_k,\mathbf{Z})=f(\mathbf {
X } _1 , \cdots
,\mathbf{X}_m,\cdots,\mathbf{X}_k,\mathbf{Z}),\ \forall\ m\in[1:k].
\eea
Using the above result we have that term (II) in (\ref{eq1precision}) satisfies
\al{
&\lim_{\rho_1\rightarrow\infty}\cdots\lim_{\rho_k\rightarrow\infty}\frac{1}{n
}\E\Big { [ } \parallel \mathbf { \tilde { X
} }
_m-g_{d_m}(F_m(\mathbf{\tilde{X}}_1,\cdots,
\mathbf { \tilde{X}
}_k , \mathbf { Z }
))\parallel^2 \Big{]}\nonumber\\
&=\lim_{\rho_1\rightarrow\infty}\cdots \lim_{\rho_k\rightarrow\infty}\frac{1}{n
}\E\Big { [ }H(\mathbf{\tilde{X}}_1,\cdots,\mathbf{\tilde{X}}_k,Z)\Big{]} \nonumber \\
&\stackrel{(a)}{=}\lim_{\rho_1\rightarrow\infty}\cdots\lim_{\rho_{k-1}
\rightarrow\infty } \frac { 1 } { n
}\E\Big { [
}H(\mathbf{\tilde{X}}_1,\cdots,\mathbf{\tilde{X}}_{k-1},\mathbf{X}_k,Z)\Big{]} \nonumber \\
&\stackrel{(b)}{=} \frac { 1 } { n
}\E\Big { [
}H(\mathbf{{X}}_1,\cdots,\mathbf{X}_k,Z)\Big{]} \nonumber \\
&\le \frac{1}{n}\E\Big{[}\parallel
\mathbf{X}_m-g_{d_m}(F_m(\mathbf{X}_1,\cdots,\mathbf{X
}_k , \mathbf { Z }
))\parallel^2\Big { ] } \le D_m
}
where $(a)$ follows from the fact that pointwise convergence of the density implies convergence
in distribution of a (measurable) function of the random variable and this implies convergence in expectation via the Dominated Convergence Theorem (see Appendix \ref{appendixH}), as we have 
from the fact that $g_m(\cdot)$ is bounded (say by $M$), 
\al{ 
& \frac{1}{n}\E \left\| {\bf X}_m - g_{d_m}\left( F_m\left( \vec{X}_1...,\vec{X}_k, \vec{Z} \right) \right)  \right\|^2 \leq \frac{2}{n} \E \left( \left\| {\bf X}_m \right\|^2 +  \left\| g_{d_m}\left( F_m\left( \vec{X}_1,...,\vec{X}_k, \vec{\tilde Z}\right) \right)  \right\|^2 \right) \nonumber \\
& \quad \quad \leq 2 \E \left( \left\| {\bf X}_m \right\|^2 +  M^2 \right) = 2 {\bf K}_{m,m} + 2 M^2 < \infty,
}
and $(b)$ follows from similarly repeating
$(a)$ by taking one limit at a time. 

Now bounding the cross term (III) in (\ref{eq1precision}), 

\bea
&&\frac{1}{n}\E\Big{[}\parallel\mathbf{X}_m-\mathbf{\tilde{X}}
_m\parallel \parallel \mathbf { \tilde { X } }
_m-g_{d_m}(F_m(\mathbf{\tilde{X}}_1,\cdots,
\mathbf { \tilde{X}
}_k, \mathbf { Z }
))\parallel \Big{]}\nonumber\\
&\le&\frac{1}{\sqrt{n}}2^{1-\rho_m}\E\Big{[}\parallel \mathbf { \tilde { X } }
_m-g_{d_m}(F_m(\mathbf{\tilde{X}}_1,\cdots,
\mathbf { \tilde{X}
}_k , \mathbf { Z }
))\parallel \Big{]} \nonumber \\
&\le&\frac{1}{\sqrt{n}}2^{1-\rho_m}\sqrt{\E\Big{[}\parallel \mathbf { \tilde { X
} }
_m-g_{d_m}(F_m(\mathbf{\tilde{X}}_1,\cdots,
\mathbf { \tilde{X}
}_k , \mathbf { Z }
))\parallel^2 \Big{]}}
\eea
and using the bound on the term (II), implies that in limit this term is bounded as $2^{1-\rho_m}\sqrt{D_m}$ which vanishes.
Hence we have proved that
\bea
\lim_{\rho_1\rightarrow\infty}\cdots\lim_{\rho_k\rightarrow\infty}\frac{1}{n}
\E\Big { [ } \parallel
\mathbf{X}_m-g_{d_m}(F_m(\mathbf{\tilde{X}}_1,\cdots,
\mathbf { \tilde{X}
}_k, \mathbf { Z }
))\parallel^2\Big { ] }&\le& \frac{1}{n}\E\Big{[}\parallel
\mathbf{X}_m-g_{d_m}(F_m(\mathbf{X}_1,\cdots,\mathbf{X
}_k, \mathbf { Z }
))\parallel^2\Big { ] } \nonumber \\
&\le& D_m.
\eea
Thus for any $\epsilon>0$, we can choose $\rho\in\N^k$, with components large enough so 
$\mathcal{C}_{\rho}$ achieves the distortion tuple, $(D_1+\epsilon,\cdots,D_k+\epsilon)$.
What is left is to show we can dispense away with random encoders. This is
argued in a standard manner by choosing the best randomizations
$\mathbf{V}_{i}$'s at respective encoders, as done in \cite{wcnoisefull}.

\section{Proof of Lemma  \ref{finiteprecisionlemma2}}
\label{appendixE}


The proof of Lemma \ref{finiteprecisionlemma2} follows similar steps to those in the proof of Lemma \ref{finiteprecisionlemma}.
We start by noticing that, by definition, if the distortion tuple $(D_1,\cdots,D_k)$ is achievable on an AWGN network, then we must have a coding scheme
$\mathcal{C}$ with block length $n$, such that,
\bea
\frac{1}{n}\E\Big{[}\parallel
\mathbf{X}_m-\hat{\mathbf{X}}_m\parallel^2\Big{]}\le D_m,\ \forall\ m=[1:k].
\eea
Using Lemma \ref{boundedoutputlemma}, without loss of generality we will suppose that
\aln{
\left\| g_{d_j} ( y_1,...,y_n) \right\|_\infty \leq M,
}
for each destination $d_j \in \D$.
We will build a randomized coding scheme $\C_\rho$, for $\rho = (\rho_1,...,\rho_N)$ by defining 
the encoding function $\tilde f_{s_m,t}$ at each source $s_m \in \mathcal{S}$ as
\aln{
\tilde f_{s_m,t} ( \vec X_m, Y_{s_m}[0:t-1]) = f_{s_m,t} (  \vec X_m, \tilde Y_{s_m}^{(\rho)}[0:t-1]),
}
and the encoding functions $\tilde f_{i,t}$ at each node $i \in \Rscr \cup \Dscr$ as
\aln{
\tilde f_{i,t} (Y_i[0:t-1]) = f_{s_m,t} ( \tilde Y_i^{(\rho)}[0:t-1]),
}
where $f_{s_m,t}$ and $f_{i,t}$ are the encoding functions of the original coding scheme $\C$ and $\tilde Y_i^{(\rho)}[t]$ is the effective received signal at node $i$ at time $t$, obtained as
\aln{
\tilde Y_i^{(\rho)}[t] = \left\lfloor Y_i[t]\right\rfloor_{\rho_i} + V_{i}^{(\rho)}[t]
}
where $V_{i}^{(\rho)}[t]$ is an i.i.d.~sequence of random variables drawn from $(-2^{-\rho_i-1},2^{-\rho_i-1})$, independent of the transmit and receive signals in the network.

Now let $\vec X = (\vec X_1,...,\vec X_k)$ be the vector of length $n k$ with the $k$ source sequences, and let $\vec Y$ be the random vector of length $n N$ corresponding to all the received signals at all nodes during the $n$ time steps in the block if code $\C$ is used. 
We also write $\vec Y = \left( \vec Y[0], ..., \vec Y[n-1] \right)$, where $\vec Y[t] = (Y_{1}[t],...,Y_{N}[t])$ is the random vector of received signals at all $N$ nodes at time $t$, for $0 \leq t \leq n-1$.
Therefore, the conditional joint density of $\vec Y$ conditioned on $\vec X = \vec x$ can be expressed as
\al{ \label{gprod} \rescnt
f_{\vec Y | \vec X} (\vec y|\vec x) & =  \prod_{t=0}^{n-1}  f_{\vec Y[t]| \vec Y[0],..., \vec Y[t-1],\vec X} \left( \left. \vec y[t] \right| \vec y[0],...,\vec y[t-1],\vec x \right).
} \rescnt
Similarly, we let $\vec {\tilde Y}^{(\rho)}$ be the vector of $n N$ effective received signals if 
code $ \C_{\rho}$ is used instead.
We also let  $\vec {\tilde Y}^{(\rho)} = \left( \vec {\tilde Y}^{(\rho)}[0], ..., \vec {\tilde Y}^{(\rho)}[n-1] \right)$, where $\vec {\tilde Y}^{(\rho)}[t] = \left({\tilde Y}^{(\rho)}_{1}[t],...,{\tilde Y}^{(\rho)}_{N}[t]\right)$.
The conditional joint density of $\vec{\tilde Y}^{(\rho)}$ conditioned on $\vec X = \vec x$ can be expressed as
\al{ \label{grprod} \rescnt
f_{\vec {\tilde Y}^{(\rho)} | \vec X} (\vec y|\vec x) & =  \prod_{t=0}^{n-1} f_{\vec{\tilde Y}^{(\rho)}[t]| \vec {\tilde Y}^{(\rho)}[0],..., \vec {\tilde Y}^{(\rho)}[t-1],\vec X} \left( \left. \vec y[t] \right| \vec y[0],...,\vec y[t-1],\vec x \right).
} \rescnt
By applying Lemma \ref{densitylemma} $N$ times, for any choice of previously received signals $\vec y[0],...,\vec y[t-1]$ and source sequences $\vec x$, we have that
\aln{
\lim_{\rho_1 \to \infty} \cdots \lim_{\rho_N \to \infty} & f_{\vec{\tilde Y}^{(\rho)}[t]| \vec {\tilde Y}^{(\rho)}[0],..., \vec {\tilde Y}^{(\rho)}[t-1],\vec X} \left( \left. \vec y[t] \right| \vec y[0],...,\vec y[t-1],\vec x \right) \\
 = & f_{\vec{Y}[t]| \vec { Y}[0],..., \vec {Y}[t-1],\vec X} \left( \left. \vec y[t] \right| \vec y[0],...,\vec y[t-1],\vec x \right),
}
for almost all $\vec y[t]$.
Therefore, (\ref{gprod}) and (\ref{grprod}) imply that
\aln{
\lim_{\rho_1 \to \infty} \cdots \lim_{\rho_N \to \infty} & \prod_{t=0}^{n-1} f_{\vec{\tilde Y}^{(\rho)}[t]| \vec {\tilde Y}^{(\rho)}[0],..., \vec {\tilde Y}^{(\rho)}[t-1],\vec X} \left( \left. \vec y[t] \right| \vec y[0],...,\vec y[t-1],\vec x \right) \\
 = & \prod_{t=0}^{n-1}  f_{\vec Y[t]| \vec Y[0],..., \vec Y[t-1],\vec X} \left( \left. \vec y[t] \right| \vec y[0],...,\vec y[t-1],\vec x \right),
}
and, in particular, we can choose a sequence $\rho[i] = (\rho_1[i],...,\rho_N[i])$, $i=1,2,...$, such that
\aln{
\lim_{i \to \infty} & \prod_{t=0}^{n-1} f_{\vec{\tilde Y}^{(\rho[i])}[t]| \vec {\tilde Y}^{(\rho[i])}[0],..., \vec {\tilde Y}^{(\rho[i])}[t-1],\vec X} \left( \left. \vec y[t] \right| \vec y[0],...,\vec y[t-1],\vec x \right) \\
 = & \prod_{t=0}^{n-1}  f_{\vec Y[t]| \vec Y[0],..., \vec Y[t-1],\vec X} \left( \left. \vec y[t] \right| \vec y[0],...,\vec y[t-1],\vec x \right),
}
%
We conclude that $f_{\vec{\tilde Y}^{(\rho[i])}|\vec X} \left( \vec y |\vec x \right) \to f_{\vec Y|\vec X} (\vec y|\vec x)$ as $i \to \infty$ for almost all $\vec y \in \R^{n N}$ and any $\vec x$.
By Scheff\'e's Theorem \cite{billingsley}, pointwise convergence of the density implies convergence in total variation.
This, in turn, implies convergence in total variation of $\left\| \vec X_m - g_{d_m}(\tilde {\vec Y}_{d_m}^{(\rho[i])}) \right\|^2$ to
 $\left\| \vec X_m - g_{d_m}({\vec Y}_{d_m}) \right\|^2$ as $i\to\infty$, which clearly implies that
 \aln{
\left\| \vec X_m - g_{d_m}(\tilde {\vec Y}_{d_m}^{(\rho[i])}) \right\|^2 \stackrel{d}{\to} 
\left\| \vec X_m - g_{d_m}({\vec Y}_{d_m}) \right\|^2.
 }
From the Dominated Convergence Theorem (Appendix \ref{appendixH}), which can be used since
\aln{
& \E \left\| \vec X_m - g_{d_m}(\tilde {\vec Y}_{d_m}^{(\rho[i])}) \right\|^2 \nonumber \\
& \quad \quad \leq 2 \E \left( \left\| {\bf X}_m \right\|^2 +  \left\| g_{d_m}(\tilde {\vec Y}_{d_m}^{(\rho[i])}) \right\|^2 \right) \nonumber \\
& \quad \quad \leq 2 \E \left( \left\| {\bf X}_m \right\|^2 +  M^2 \right) \nonumber \\
& \quad \quad = 2 {\bf K}_{m,m} + 2 M^2 < \infty.
}
We conclude that
\aln{
\lim_{i \to \infty} \frac1n \E \left[ \left. \left\| \vec X_m - g_{d_m}(\tilde {\vec Y}_{d_m}^{(\rho[i])}) \right\|^2 \right| \vec X = \vec x \right] = 
\frac1n \E \left[ \left. \left\| \vec X_m - g_{d_m}({\vec Y}_{d_m}) \right\|^2 \right| \vec X = \vec x \right]
}
for any fixed $\vec X= \vec x$, and for each decoder $d_m$.
Thus, the random variable $\E \left[ \left. \left\| \vec X_m - g_{d_m}(\tilde {\vec Y}_{d_m}^{(\rho[i])}) \right\|^2 \right| \vec X \right]$ converges surely to $\E \left[ \left. \left\| \vec X_m - g_{d_m}({\vec Y}_{d_m}) \right\|^2 \right| \vec X \right]$.
Finally, by noticing that
\aln{
\E \left[ \left. \left\| \vec X_m - g_{d_m}(\tilde {\vec Y}_{d_m}^{(\rho[i])}) \right\|^2 \right| \vec X  \right]
\leq 2 \E \left[ \left. \left\| \vec X_m \right\|^2 \right| \vec X  \right] + 2 M^2
}
and that
\aln{
\E \left[ \E \left[ \left. \left\| \vec X_m \right\|^2 \right| \vec X \right] + 2 M^2 \right] = 2 {\vec K}_{m,m} + 2 M^2,
}
a second application of the Dominated Convergence Theorem implies that
\aln{
& \lim_{i \to \infty} \frac1n \E \left[ \left\| \vec X_m - g_{d_m}(\tilde {\vec Y}_{d_m}^{(\rho[i])}) \right\|^2  \right] = 
\lim_{i \to \infty} \frac1n \E\left[ \E \left[ \left. \left\| \vec X_m - g_{d_m}(\tilde {\vec Y}_{d_m}^{(\rho[i])}) \right\|^2 \right| \vec X  \right] \right] \\
& \quad \quad = 
\frac1n \E \left[ \E \left[ \left. \left\| \vec X_m - g_{d_m}({\vec Y}_{d_m}) \right\|^2 \right| \vec X  \right] \right]
= \frac1n \E \left[ \left\| \vec X_m - g_{d_m}({\vec Y}_{d_m}) \right\|^2  \right]
}
for each decoder $d_m$.
Thus we can choose large enough $i$ so that performance of
$\mathcal{C}_{\rho[i]}$ will be arbitrarily close to that of $\mathcal{C}$.
What is left is to show we can dispense away with random encoders. This is
argued in a standard manner by choosing the best randomizations
$\mathbf{V}_{i}$'s at respective encoders. 

\section{Proof of Lemma  \ref{continuouslemma}}
\label{appendixF}
Denote the set $\mathscr{S}(\rho)=\{x\in\R^a:2^\rho x\in\Z^a\}$, where $\Z$ is the set of integers. 
Note that the function in the theorem can take values $f(y)$ where $y\in\mathscr{S}(\rho)$. Now for each $y\in\mathscr{S}(\rho)$, define the set $S(y)=\{x\in\R^a : x\neq y, \lfloor x\rfloor_{\rho}=y \}$, which are disjoint for different values of $y\in\mathscr{S}(\rho)$ and cover the whole space $\R^a$.  Since $f$ takes a constant value in each of the sets $S(\cdot)$, the only regions of discontinuity are the boundaries of these regions. But these boundaries are disjoint bounded rectangles each of which has Lebesgue measure zero, implying the total region of discontinuity has zero measure.
Thus $f$ is locally constant almost-everywhere (and hence continuous).

\section{Dominated Convergence Theorem}
\label{appendixH}
We require the following version of the Dominated Convergence Theorem.


\begin{theorem}
Suppose we have a sequence of random vectors ${\bf Z}_n \in \R^a$ converging weakly to ${\bf Z}$, and two almost-everywhere continuous functions $f, g : \R^a \to \R$ such that $0 \leq f \leq g$.
Then, if $E[g({\bf Z}_n)] = E[g({\bf Z})] = c < \infty$ for all $n$, we have $\lim_{n \to \infty} E[f({\bf Z}_n)] = E[f({\bf Z})]$.
\end{theorem}

\begin{proof}
If we let $X_n = f({\bf Z}_n)$, $Y_n = g({\bf Z}_n)$, $X = f({\bf Z})$ and $Y = g({\bf Z})$, from the almost everywhere continuity of the functions, we have $X_n\stackrel{d}{\to} X$ and $Y_n \stackrel{d}{\to} Y$. From Theorem 25.11 in \cite{billingsley}, we have that
\aln{
E[X] \leq \liminf_{n \to \infty} E[X_n].
}

Note that, $Y_n-X_n=g({\bf Z_n})-f({\bf Z_n})$ is an almost everywhere continuous function of ${\bf Z_n}$, hence the sequence of random variables $Y_n-X_n$, converges weakly to $Y-X$. Therefore, since $Y_n - X_n \geq 0$, a second application of Theorem 25.11 yields
\aln{
c - E[X] & = E[Y-X] \leq \liminf_{n \to \infty} E[Y_n - X_n] \\
& =  \liminf_{n \to \infty} c - E[X_n] = c - \limsup_{n \to \infty} E[X_n].
}
Combining both inequalities, we obtain
\aln{
\limsup_{n \to \infty} E[X_n] \leq E[X] \leq \liminf_{n \to \infty} E[X_n],
}
which implies that  $\lim_{n \to \infty} E[X_n] = E[X]$.
\end{proof}

\end{document}